\documentclass[aps,prr,final,twocolumn,superscriptaddress]{revtex4-2}
\usepackage[colorlinks=true, allcolors=blue]{hyperref}
\usepackage{graphicx} 
\usepackage{braket}
\usepackage{tikz}
\usepackage{amsmath,amsthm,amssymb}
\usetikzlibrary{arrows}
\usepackage{bm}
\usepackage{bbold}
\usepackage{amsfonts}
\usepackage{mathtools}
\usepackage{tikz-cd}
\usepackage{verbatim}
\usepackage{graphicx}

\usepackage{caption}
\usepackage{subcaption}
\usepackage{ragged2e}
\DeclareCaptionJustification{justified}{\justifying}
\usepackage{thmtools, thm-restate}
\newtheorem{theorem}{Theorem}
\newtheorem{corollary}[theorem]{Corollary}
\newtheorem{lemma}[theorem]{Lemma}

\newtheorem{remark}[theorem]{Remark}
\usepackage{adjustbox}

\newcommand{\stot}{S_\text{tot}^z}
\newcommand{\sutwo}{\operatorname{SU}(2)}

\DeclareMathOperator{\nullity}{nullity}
\newcommand{\I}{\mathrm{i}}
\newcommand{\refEq}[1]{Eq.~(\ref{#1})}
\newcommand{\refEqTable}[1]{Eq.(\ref{#1})}
\newcommand{\refTab}[1]{Tab.~(\ref{#1})}
\newcommand{\refFig}[1]{Fig.~\ref{#1}}
\newcommand{\refApp}[1]{App.~\ref{#1}}

\newcommand{\annotateEquation}[1]{\\ [-1.3ex]%
    \scalebox{0.75}{\tiny \refEqTable{#1} \par}}

\definecolor{cbRed}{RGB}{204, 51, 17}    
\definecolor{cbGreen}{RGB}{0, 153, 136}
\definecolor{cbBlue}{RGB}{0, 119, 187}

\newcommand{\commensurate}[1]{%
  \begin{minipage}[t]{1.0\linewidth}%
  \setlength{\fboxsep}{1.5pt}
    \textcolor{cbRed}{\fbox{#1}} \annotateEquation{eq:degen-thm-summary}%
  \end{minipage}%
}

\newcommand{\incommensurateNOddQOne}[1]{%
  \begin{minipage}[t]{1.0\linewidth}%
    \bm{\textcolor{cbGreen}{#1}} \annotateEquation{eq:degen-thm-non-summary}%
  \end{minipage}%
}

\newcommand{\incommensurateNOddQMinusOne}[1]{%
  \begin{minipage}[t]{1.0\linewidth}%
    \bm{\textcolor{cbGreen}{#1}} \annotateEquation{eq:degen-thm-non-summary}%
  \end{minipage}%
}

\newcommand{\incommensurateNEven}[1]{%
  \begin{minipage}[t]{1.0\linewidth}%
    \bm{\textcolor{cbGreen}{#1}} \annotateEquation{eq:degen-thm-non-summary}%
  \end{minipage}%
}

\newcommand{\specialcaseXX}[1]{%
  \begin{minipage}[t]{1.0\linewidth}%
    \textcolor{cbBlue}{#1} \annotateEquation{eq:degen-xx}%
  \end{minipage}%
}

\newcommand{\specialcaseXXX}[1]{%
  \begin{minipage}[t]{1.0\linewidth}%
    \textcolor{cbBlue}{#1} \annotateEquation{eq:xxx}%
  \end{minipage}%
}
\newcommand{\specialcaseXXMinusX}[1]{%
  \begin{minipage}[t]{1.0\linewidth}%
    \textcolor{cbBlue}{#1} \annotateEquation{eq:degen-minusone}%
  \end{minipage}%
}

\newcommand{\invalid}[1]{%
  \begin{minipage}[t]{1.0\linewidth}%
    \textcolor{cbBlue}{ INVALID }\annotateEquation{}%
  \end{minipage}%
}
\newcommand{\badInput}[1]{%
  \begin{minipage}[t]{1.0\linewidth}%
    \textcolor{cbBlue}{ BAD INPUT } \annotateEquation{}%
  \end{minipage}%
}

\tikzset{
    twistColor/.style={cbRed},
    statesNode/.style={
    draw, thick, twistColor, rectangle, minimum width=1cm, minimum height=0.6cm, double, double distance=1pt
    },
    connection/.style={
        dashed, twistColor, thick, >={Stealth[scale=0.8]}, 
    },
    connectionDifferent/.style={
        -{Stealth[scale=0.8]}, cbGreen, densely dotted, line width=2pt,
        shorten >=0.05cm,
        shorten <=0.05cm,
        ->,
        to path={([yshift=0.1cm]\tikztostart) -- ([yshift=-0.125cm]\tikztotarget) \tikztonodes}
    },
    connectionSame/.style={
        cbBlue, line width=1.5pt,
        shorten >=0.05cm,
        shorten <=0.05cm,
        >={Stealth[scale=0.8]}, 
        ->,
        to path={([yshift=0.1cm]\tikztostart) -- ([yshift=-0.125cm]\tikztotarget) \tikztonodes}
    },
}

\begin{document}
\title{Hidden Twisted Sectors and Exponential Degeneracy \texorpdfstring{\\}{} in Root-of-Unity XXZ Heisenberg Chains}
\author{Yongao Hu}
\affiliation{Department of Physics, Massachusetts Institute of Technology, 77 Massachusetts Avenue,  Cambridge, Massachusetts, USA}
\author{Felix Gerken}
\affiliation{I. Institut für Theoretische Physik, Universität Hamburg, Hamburg, Germany}
\author{Thore Posske}
\affiliation{I. Institut für Theoretische Physik, Universität Hamburg, Hamburg, Germany}
\date{\today} 
\begin{abstract}
Recently, product states have been identified as simple-structured eigenstates of XXZ Heisenberg spin models in arbitrary dimensions, occurring at anisotropy values corresponding to certain roots of unity. Yet, the product states typically only span parts of a larger degenerate eigenspace. 
Here, we classify this eigenspace in the one-dimensional periodic XXZ chain at all roots of unity $q$, where $q^2$ is an $\ell$-th primitive root of unity. 
For commensurate chain lengths $N$ with $q^N=1$, we prove that the minimal degeneracy is $2^{N/\ell}\ell$ using the representation theory of the affine Temperley-Lieb (aTL) algebra.
For the incommensurate case, we derive analogous exponential lower bounds of $2^{2\lfloor\frac{N}{2\ell}\rfloor+1}$ if $N$ is even and $2^{2\lfloor \frac{N}{2\ell}+\frac{1}{2}\rfloor}$ if $N$ is odd and $q^\ell=1$.
Our proof employs the morphisms between aTL modules discovered by Pinet and Saint-Aubin \cite{Pinet_2022} and emphasizes the importance of exact sequences and hidden twisted boundary condition sectors that mediate the degeneracy. In the case of commensurate chain lengths, we connect to the Fabricius-McCoy string construction of all Bethe roots of the degenerate subspace, which previously uncovered parts of our results. 
We corroborate our results numerically and demonstrate that the lower bound is saturated for chain lengths $N\leq20$. Our work demonstrates for a concrete system how the interplay of the Bethe ansatz, aTL representation theory, and twisted boundary conditions explains degeneracy connected to long-lived product states \cite{Jepsen:2021gqv}, stimulating research towards generalization to higher dimensions. Exponential degeneracy could boost applications of spin chains as quantum sensors.
\end{abstract}
\maketitle

\section{Introduction}

Spin systems are archetypal quantum many-body systems that describe magnetic properties of materials, in particular their phase transitions and critical points \cite{heisenberg_zur_1928,BruusFlensberg2004,karle2021}, model special setups in cold atoms \cite{jepsen2020SpinTransportTunable,Jepsen:2021gqv} and quantum computers \cite{vandyke2022PreparingExactEigenstates,jaderberg2025VariationalPreparationNormal,lutz2025AdiabaticQuantumState,komelj2025QuantumComputingMagneticskyrmionpatterns}, and serve as theoretical laboratories for non-perturbative phenomena in high-energy physics \cite{Maldacena:1997re,PRXQuantum.4.027001,Beisert:2010jr,Minahan:2002ve,Faddeev:1994zg}.
The anisotropic spin-1/2 Heisenberg chain \cite{heisenberg_zur_1928,BAXTER1972323,deguchi_xxz_2007,wruff2013xxz,tolmay2015xxz}, where each spin interacts by exchange interaction with its nearest neighbors, is one of the simplest nontrivial models \cite{tang2026TopologicalDualXXZa}. Its eigenproblem can be solved exactly with the coordinate Bethe ansatz \cite{bethe_zur_1931,orbach1958LinearAntiferromagneticChain,YangYang1966,yang1967}, making it a testing ground for understanding quantum magnetism, entanglement, and emergent phenomena in one dimension~\cite{Korepin:1993kvr,Klumper2004, YangYang1966}, with extensions to higher-dimensional lattices~\cite{Lin1989,Barabanov1994,Waldtmann:1998mea}). 
The one-dimensional spin-$1/2$ XXZ Hamiltonian with periodic boundary condition is \cite{heisenberg_zur_1928}
\begin{equation}
\label{eq:h_1d}
H_\text{XXZ} =\sum_{i=1}^N\left(S_i^xS_{i+1}^x+S_i^yS_{i+1}^y+\Delta S_i^zS_{i+1}^z\right),
\end{equation}
where $\Delta$ is the Heisenberg anisotropy in the $z$-direction, $N$ is the total number of spins, and ${S}_i$'s are the spin operators at site $i$ with the periodic boundary condition \cite{karbach1998introduction,jianglecture}
${S}_{i+N}\equiv {S}_i.$
In the gapless regime where $|\Delta| \leq 1$, the anisotropy can alternatively be described by two other parameters  
$\gamma = \arccos{\Delta}$
or 
$q=e^{\I \gamma}.$
Special phenomena and enhanced degeneracy emerge when the anisotropy corresponds to $q$ being a root of unity \cite{PASQUIER1990523, fabricius_bethes_2001,deguchi_sl_2_2000,deguchi2002ConstructionMissingEigenvectors}. We label $q^2$ to be an $\ell$-th primitive root of unity ($\ell$ is the smallest positive integer such that $q^{2\ell}=1$).
Physically, the length $\ell$ is the minimal distance between collinear spins of the helix at that anisotropy, i.e., half the period or the full period for even and odd $N$, respectively. 
At these anisotropies, the standard Bethe ansatz breaks down due to zero-energy excitations stemming from string solutions \cite{fabricius_bethes_2001,fabricius2001_2}. The enhanced degeneracies cannot be explained by translation and spin-flip symmetry alone. Instead, they are connected to the representation theory of quantum groups \cite{quantumgroupklimyk,woit2017quantum,PASQUIER1990523,Pinet_2022,faddeev1996algebraic,Chari:1994pz,tjerk1991hopf,Jimbo:1985zk,gainutdinov2015CountingSolutionsBethe} and affine Temperley-Lieb (aTL) algebras \cite{Aufgebauer_2010,nepomechie2016UniversalBetheAnsatz,PASQUIER1990523,Pinet_2022}. 

A remarkable feature at special roots of unity are product eigenstates in spin models of diverse dimensions \cite{gerken_all_2023,popkov_phantom_2021,Cerezo2016,Jepsen:2021gqv}.
These zero-entanglement states in an otherwise highly correlated quantum system have implications for quantum thermalization \cite{Essler:2016ufo,Prosen:2013woz} and can be interpreted as many-body quantum scars \cite{Zhang:2023kxc,chandran_quantum_2023,Pizzi:2024urq,Jepsen:2021gqv}. 
In spin chains, the product eigenstates form helices that close on themselves commensurately after an integer number of windings, meaning $q^N=1$, and appear as phantom helices in chiral reformulations of the Bethe ansatz~\cite{popkov_phantom_2021,zhang_generalized_2023,Zhang_2021,Zhang:2023kxc}.
They all have the same energy as the trivial fully-polarized states $\ket{\uparrow\dots\uparrow}$ and $\ket{\downarrow\dots\downarrow}$, generally in the middle of the spectrum~\cite{gerken_all_2023} but in special cases also in the ground state \cite{Changlani2018,Batista2012}.
Yet, there is a large excess degeneracy at the product state energy, not attributed to the product eigenstates alone nor given by common symmetries of the system.
The general underlying mechanism is unknown, and insight is given only from special cases. 
For one-dimensional spin ladders, the excess degeneracy is associated with anyonic condensates that coexist with product state helices \cite{Batista2009,Batista2012}.
For diamond chains, the degeneracy stems from local symmetries \cite{Sutherland1983}.
Regarding higher-dimensional systems, the square lattice with zig-zag boundaries \cite{gerken_all_2023,Miao2025} shows excess degeneracy, which could be numerically resolved, but the ground space degeneracy of the kagome lattice on a torus is fully explained by product eigenstates and any excess is absent \cite{Changlani2018}. In this case, there is a noteworthy proximity of the product state degeneracy to a spin liquid phase \cite{Changlani2018}.

For the linear XXZ chain, Fabricius-McCoy (FM) string solutions~\cite{fabricius_bethes_2001,fabricius2001_2} of the Bethe ansatz equations at roots of unity can account for many of the degenerate states \cite{braak2001,deguchi_sl_2_2000,deguchi2002ConstructionMissingEigenvectors,deguchi_xxz_2007}.
As recently demonstrated by Ref.~\cite{Pinet_2022}, this construction is intimately connected to the representation theory of the aTL algebra at roots of unity, where spin chains are treated as aTL-modules and the morphisms between these modules alternatively explain the degeneracies. 
Explicit degeneracy formulas using FM strings exist for special cases \cite{deguchi_xxz_2007,deguchi2002ConstructionMissingEigenvectors,deguchi_sl_2_2000}, notably for the commensurate case with even $N/\ell$ \cite{deguchi2002ConstructionMissingEigenvectors}. However, these formulas rely on the FM-string hypothesis, which implies they are not rigorously proven in magnetization sectors where $\stot \not\equiv 0\mod \ell$. Likewise, exact degeneracies of the aTL-morphisms~\cite{Pinet_2022} have, to the best of our knowledge, not been calculated. 
A systematic classification of the eigenspace at product state energy thus has remained an open problem, particularly at roots of unity incommensurate with the chain length, where $q^N\neq1$. 

In this article, we provide the complete classification of the degenerate eigenspace at the product state energy in the one-dimensional periodic XXZ model at roots of unity. Our main contributions are as follows:
First, for the commensurate chains ($q^N=1$) and $q^2$ an $\ell$-th primitive root of unity with $\ell>2$,  we prove that the degeneracy $g$ at the product state energy satisfies
\begin{align}
\label{eq:degen-thm-summary}
g \geq 2^{N/\ell}\ell,
\end{align}
without relying on the FM-string hypothesis, using the representation theory of the aTL algebra and in particular the morphisms between modules constructed in Ref.~\cite{Pinet_2022}.
The product eigenstates themselves contribute a degeneracy of $2N$ for $\Delta\neq\pm1$.
We discover that the degeneracy is mediated through a hidden half-sequence of XXZ chains with {twisted} periodic boundary conditions, where instead of periodic boundary conditions 
one sets $S_{N+1}^\pm=e^{i\phi}S_1^\pm, \; S_{N+1}^z=\sigma_1^z$ for some non-zero $\phi$.
Second, we extend the construction for the product states to the general incommensurate case where $N$ is not necessarily divisible by $\ell$.
We prove the lower bounds on the degeneracy for these cases is
\begin{equation}
    \label{eq:degen-thm-non-summary}
        g \geq \begin{cases}
            2^{2\lfloor\frac{N}{2\ell}+\frac{1}{2}\rfloor}\quad &\text{for odd } N\text{ and }q^{\ell}=1,\\
            2\quad&\text{for odd } N\text{ and }q^{\ell}=-1,\\
            2^{2\lfloor\frac{N}{2\ell}\rfloor+1}\quad &\text{for even } N.
        \end{cases}
\end{equation}
This is a generalization beyond the known range and demonstrates that the morphisms in the aTL algebra found in Ref.~\cite{Pinet_2022} are able to link the degeneracy between commensurate and incommensurate periodic chains through the hidden half-sequences of twisted chains. 
Moreover, we demonstrate numerically that the lower bound is saturated for all $N\leq20$, see \refTab{tab:degeneracies}. To emphasize the physical picture and connection to the Bethe ansatz, we explicitly demonstrate how the aTL morphisms are tightly connected to the perspective of FM string structures. 

The article is organized as follows. Section~\ref{sec:review} offers a review of the XXZ model, the product state construction, and spin chains as modules over the aTL algebra. Section~\ref{sec:structure} offers a motivation of the degenerate subspace structure using the FM string construction. Finally, section~\ref{sec:algebra} connects the structure to the intertwiners between aTL modules and rigorously proves the lower bounds of degeneracies for both commensurate and incommensurate chains.

\begin{table*}
\bgroup \def\arraystretch{1.1} 
\setlength{\tabcolsep}{0.2em} 
\begin{tabular}{|c|c|*{18}{p{0.041\linewidth}|}p{0.051\linewidth}|}
\hline & & \multicolumn{19}{c|}{$N$} \\ 
\cline{3-21}
$ \frac{2 \pi }{\gamma} $ & $\ell$ & 
\multicolumn{1}{c|}{$2$} & \multicolumn{1}{c|}{$3$} & \multicolumn{1}{c|}{$4$} & 
\multicolumn{1}{c|}{$5$} & \multicolumn{1}{c|}{$6$} & \multicolumn{1}{c|}{$7$} & 
\multicolumn{1}{c|}{$8$} & \multicolumn{1}{c|}{$9$} & \multicolumn{1}{c|}{$10$} & 
\multicolumn{1}{c|}{$11$} & \multicolumn{1}{c|}{$12$} & \multicolumn{1}{c|}{$13$} & 
\multicolumn{1}{c|}{$14$} & \multicolumn{1}{c|}{$15$} & \multicolumn{1}{c|}{$16$} & 
\multicolumn{1}{c|}{$17$} & \multicolumn{1}{c|}{$18$} & \multicolumn{1}{c|}{$19$} & 
\multicolumn{1}{c|}{$20$} \\
\hline$1$ & $-$&\specialcaseXXX{$3$}&\specialcaseXXX{$4$}&\specialcaseXXX{$5$}&\specialcaseXXX{$6$}&\specialcaseXXX{$7$}&\specialcaseXXX{$8$}&\specialcaseXXX{$9$}&\specialcaseXXX{$10$}&\specialcaseXXX{$11$}&\specialcaseXXX{$12$}&\specialcaseXXX{$13$}&\specialcaseXXX{$14$}&\specialcaseXXX{$15$}&\specialcaseXXX{$16$}&\specialcaseXXX{$17$}&\specialcaseXXX{$18$}&\specialcaseXXX{$19$}&\specialcaseXXX{$20$}&\specialcaseXXX{$21$}\\ 
$2$ & $-$&\specialcaseXXMinusX{$3$}&\specialcaseXXMinusX{$2$}&\specialcaseXXMinusX{$5$}&\specialcaseXXMinusX{$2$}&\specialcaseXXMinusX{$7$}&\specialcaseXXMinusX{$2$}&\specialcaseXXMinusX{$9$}&\specialcaseXXMinusX{$2$}&\specialcaseXXMinusX{$11$}&\specialcaseXXMinusX{$2$}&\specialcaseXXMinusX{$13$}&\specialcaseXXMinusX{$2$}&\specialcaseXXMinusX{$15$}&\specialcaseXXMinusX{$2$}&\specialcaseXXMinusX{$17$}&\specialcaseXXMinusX{$2$}&\specialcaseXXMinusX{$19$}&\specialcaseXXMinusX{$2$}&\specialcaseXXMinusX{$21$}\\ 
$3$ & $3$& &\commensurate{$6$}&\incommensurateNEven{$2$}&\incommensurateNOddQOne{$4$}&\commensurate{$12$}&\incommensurateNOddQOne{$4$}&\incommensurateNEven{$8$}&\commensurate{$24$}&\incommensurateNEven{$8$}&\incommensurateNOddQOne{$16$}&\commensurate{$48$}&\incommensurateNOddQOne{$16$}&\incommensurateNEven{$32$}&\commensurate{$96$}&\incommensurateNEven{$32$}&\incommensurateNOddQOne{$64$}&\commensurate{$192$}&\incommensurateNOddQOne{$64$}&\incommensurateNEven{$128$}\\ 
$4$ & $2$& & &\specialcaseXX{$10$}&\specialcaseXX{$2$}&\specialcaseXX{$14$}&\specialcaseXX{$2$}&\specialcaseXX{$60$}&\specialcaseXX{$20$}&\specialcaseXX{$74$}&\specialcaseXX{$2$}&\specialcaseXX{$386$}&\specialcaseXX{$2$}&\specialcaseXX{$434$}&\specialcaseXX{$346$}&\specialcaseXX{$2160$}&\specialcaseXX{$2$}&\specialcaseXX{$6124$}&\specialcaseXX{$2$}&\specialcaseXX{$13106$}\\ 
$5$ & $5$& & & &\commensurate{$10$}&\incommensurateNEven{$2$}&\incommensurateNOddQOne{$4$}&\incommensurateNEven{$2$}&\incommensurateNOddQOne{$4$}&\commensurate{$20$}&\incommensurateNOddQOne{$4$}&\incommensurateNEven{$8$}&\incommensurateNOddQOne{$4$}&\incommensurateNEven{$8$}&\commensurate{$40$}&\incommensurateNEven{$8$}&\incommensurateNOddQOne{$16$}&\incommensurateNEven{$8$}&\incommensurateNOddQOne{$16$}&\commensurate{$80$}\\ 
$6$ & $3$& & & & &\commensurate{$12$}&\incommensurateNOddQMinusOne{$2$}&\incommensurateNEven{$8$}&\incommensurateNOddQMinusOne{$2$}&\incommensurateNEven{$8$}&\incommensurateNOddQMinusOne{$2$}&\commensurate{$48$}&\incommensurateNOddQMinusOne{$2$}&\incommensurateNEven{$32$}&\incommensurateNOddQMinusOne{$2$}&\incommensurateNEven{$32$}&\incommensurateNOddQMinusOne{$2$}&\commensurate{$192$}&\incommensurateNOddQMinusOne{$2$}&\incommensurateNEven{$128$}\\ 
$7$ & $7$& & & & & &\commensurate{$14$}&\incommensurateNEven{$2$}&\incommensurateNOddQOne{$4$}&\incommensurateNEven{$2$}&\incommensurateNOddQOne{$4$}&\incommensurateNEven{$2$}&\incommensurateNOddQOne{$4$}&\commensurate{$28$}&\incommensurateNOddQOne{$4$}&\incommensurateNEven{$8$}&\incommensurateNOddQOne{$4$}&\incommensurateNEven{$8$}&\incommensurateNOddQOne{$4$}&\incommensurateNEven{$8$}\\ 
$8$ & $4$& & & & & & &\commensurate{$16$}&\incommensurateNOddQMinusOne{$2$}&\incommensurateNEven{$8$}&\incommensurateNOddQMinusOne{$2$}&\incommensurateNEven{$8$}&\incommensurateNOddQMinusOne{$2$}&\incommensurateNEven{$8$}&\incommensurateNOddQMinusOne{$2$}&\commensurate{$64$}&\incommensurateNOddQMinusOne{$2$}&\incommensurateNEven{$32$}&\incommensurateNOddQMinusOne{$2$}&\incommensurateNEven{$32$}\\ 
$9$ & $9$& & & & & & & &\commensurate{$18$}&\incommensurateNEven{$2$}&\incommensurateNOddQOne{$4$}&\incommensurateNEven{$2$}&\incommensurateNOddQOne{$4$}&\incommensurateNEven{$2$}&\incommensurateNOddQOne{$4$}&\incommensurateNEven{$2$}&\incommensurateNOddQOne{$4$}&\commensurate{$36$}&\incommensurateNOddQOne{$4$}&\incommensurateNEven{$8$}\\ 
$10$ & $5$& & & & & & & & &\commensurate{$20$}&\incommensurateNOddQMinusOne{$2$}&\incommensurateNEven{$8$}&\incommensurateNOddQMinusOne{$2$}&\incommensurateNEven{$8$}&\incommensurateNOddQMinusOne{$2$}&\incommensurateNEven{$8$}&\incommensurateNOddQMinusOne{$2$}&\incommensurateNEven{$8$}&\incommensurateNOddQMinusOne{$2$}&\commensurate{$80$}\\ 
$11$ & $11$& & & & & & & & & &\commensurate{$22$}&\incommensurateNEven{$2$}&\incommensurateNOddQOne{$4$}&\incommensurateNEven{$2$}&\incommensurateNOddQOne{$4$}&\incommensurateNEven{$2$}&\incommensurateNOddQOne{$4$}&\incommensurateNEven{$2$}&\incommensurateNOddQOne{$4$}&\incommensurateNEven{$2$}\\ 
$12$ & $6$& & & & & & & & & & &\commensurate{$24$}&\incommensurateNOddQMinusOne{$2$}&\incommensurateNEven{$8$}&\incommensurateNOddQMinusOne{$2$}&\incommensurateNEven{$8$}&\incommensurateNOddQMinusOne{$2$}&\incommensurateNEven{$8$}&\incommensurateNOddQMinusOne{$2$}&\incommensurateNEven{$8$}\\ 
$13$ & $13$& & & & & & & & & & & &\commensurate{$26$}&\incommensurateNEven{$2$}&\incommensurateNOddQOne{$4$}&\incommensurateNEven{$2$}&\incommensurateNOddQOne{$4$}&\incommensurateNEven{$2$}&\incommensurateNOddQOne{$4$}&\incommensurateNEven{$2$}\\ 
$14$ & $7$& & & & & & & & & & & & &\commensurate{$28$}&\incommensurateNOddQMinusOne{$2$}&\incommensurateNEven{$8$}&\incommensurateNOddQMinusOne{$2$}&\incommensurateNEven{$8$}&\incommensurateNOddQMinusOne{$2$}&\incommensurateNEven{$8$}\\ 
$15$ & $15$& & & & & & & & & & & & & &\commensurate{$30$}&\incommensurateNEven{$2$}&\incommensurateNOddQOne{$4$}&\incommensurateNEven{$2$}&\incommensurateNOddQOne{$4$}&\incommensurateNEven{$2$}\\ 
$16$ & $8$& & & & & & & & & & & & & & &\commensurate{$32$}&\incommensurateNOddQMinusOne{$2$}&\incommensurateNEven{$8$}&\incommensurateNOddQMinusOne{$2$}&\incommensurateNEven{$8$}\\ 
$17$ & $17$& & & & & & & & & & & & & & & &\commensurate{$34$}&\incommensurateNEven{$2$}&\incommensurateNOddQOne{$4$}&\incommensurateNEven{$2$}\\ 
$18$ & $9$& & & & & & & & & & & & & & & & &\commensurate{$36$}&\incommensurateNOddQMinusOne{$2$}&\incommensurateNEven{$8$}\\ 
$19$ & $19$& & & & & & & & & & & & & & & & & &\commensurate{$38$}&\incommensurateNEven{$2$}\\ 
$20$ & $10$& & & & & & & & & & & & & & & & & & &\commensurate{$40$}\\ 
\hline 
\end{tabular}
\egroup

\caption{\label{tab:degeneracies} The numerically obtained degeneracies at product eigenstate energy assume exactly the lower bound for commensurate chains (red, boxed), incommensurate chains (cyan, bold), and  special cases (XXX, XX(-X), and XX chains) (blue) for all roots of unity and chain lengths $N\leq20$, see  Theorem~\ref{thm:degen}, Corollary~\ref{cor:degen-thm-non}, Remark~\ref{rem:degen-thm}, and \refApp{sec:numerics}.}
\end{table*}

\section{Background}
\label{sec:review}

We first review the 1D periodic spin-$1/2$ XXZ Heisenberg chain and its product eigenstates. 
The XXZ Hamiltonian of chain size $N$ acts on the Hilbert space $\mathcal{H}=\left(\mathbb{C}^2\right)^{\otimes N}$
of dimension $\dim{\mathcal{H}}= 2^N$ \cite{faddeev1996algebraic,karbach1998introduction}. The naive basis of $\mathcal{H}$ is $\{\ket{\sigma_1\sigma_2\cdots\sigma_N}\}$, where $\sigma_i \in \{\uparrow,\downarrow\}$ represents the spin of the ${i}^\text{th}$ site. The XXZ Hamiltonian commutes with the total $z$-spin operator $\stot=\sum_{i=1}^N S_i^z$, allowing one to block-diagonalize the Hamiltonian into different sectors $\mathcal{H}_M$, where $M$ denotes the number of spin-up sites \cite{YangYang1966,wruff2013xxz, jianglecture}. 
The 1D XXZ model is solvable by the Bethe Ansatz \cite{bethe_zur_1931, orbach1958LinearAntiferromagneticChain,YangYang1966}. The solution, the Bethe states, are energy and $\stot$ eigenstates of the system parameterized by rapidities $\{v_1,\dots,v_M\}$ that are solutions to the Bethe Ansatz Equations (BAE) for the XXZ model \cite{bethe_zur_1931}:
\begin{equation}
\label{eq:xxz_bae}
\left(\frac{\sinh{(v_j+i\gamma/2)}}{\sinh{(v_j-i\gamma/2)}}\right)^N=\prod_{k\neq j}^M\frac{\sinh{(v_j-v_k+i\gamma)}}{\sinh{(v_j-v_k-i\gamma)}}.
\end{equation}
Because of its relevance to the FM string construction, we review the details of the coordinate Bethe ansatz in Appendix~\ref{sec:transfer_matrix}. The BAEs form a nonlinear system, notoriously difficult to solve \cite{hagemans2007dynamics}, especially near or at the roots of unity, where the system can be singular and accommodates enhanced degeneracy \cite{fabricius_bethes_2001}. 

The product eigenstates of the periodic XXZ Heisenberg chain are labeled by local polar and azimuthal angles \cite{gerken_all_2023,zhang_generalized_2023,Cerezo2016}, 
\begin{align}
\label{eq:product_ansatz-2}
\ket{\psi} = \ket{\theta_1,\varphi_1}\otimes\ket{\theta_2,\varphi_2}\otimes\cdots\otimes\ket{\theta_N,\varphi_N},
\end{align}
with $\ket{\theta,\varphi} = \left( \cos{\theta} \ket{\uparrow} + e^{\I\varphi} \sin(\theta)\ket{\downarrow} \right) / \sqrt{2}$. The adjacent spins needs to satisfy
\begin{equation}
\label{eq:num_of_dege} 
\theta_i = \theta_j, \quad \varphi_i - \varphi_j = \gamma.
\end{equation}
There are two trivial product eigenstates at any value of $\Delta$ are $\ket{\uparrow\cdots\uparrow}$ and $\ket{\downarrow\cdots\downarrow}$.
In the gapless phase where $|\Delta| \leq 1$, nontrivial helix product eigenstates exist when winding angles add up to an integer multiple of $2\pi$, such that $q^N=e^{\I \gamma N} = 1$ \cite{deguchi_sl_2_2000}. 
For brevity, we refer to these roots of unity as {commensurate roots of unity}.
All product eigenstates have the same energy 
\begin{align}
\label{eq:product_state_energy}
\varepsilon = \Delta N / 4,
\end{align}
assuming $\hbar=1$, which we denote as the {product state energy}. Be aware that the nontrivial product eigenstates are not $S^z_{\text{tot}}$ eigenstates, having components in every $S^z_{\text{tot}}$, see \refEq{eq:product_ansatz-2}. 
We label the degenerate subspace spanned by the product eigenstates as $\mathcal{P}$. The degeneracy  of $\mathcal{P}$ at $\Delta\neq 1$ can be calculated \cite{gerken_all_2023} by counting the two fully polarized states
and, in case of commensurate roots of unity, adding the linearly independent projections of nontrivial product eigenstates with opposite helicity, for which there are two for each of the $N-1$ remaining magnetization sectors. This yields
\begin{equation}
    \label{eq:product_state_degeneracy}
    \dim \mathcal P= \begin{cases}
        2N \ &\text{commensurate case,}
        \\
        2 \ &\text{incommensurate case.}
    \end{cases}
\end{equation}
However, numerical diagonalization reveals additional non-product states at the same energy level that grows exponentially with $N$, see \refTab{tab:degeneracies}. We label the entire degenerate subspace at the product state $\mathcal{D}$, where $\mathcal{P}\subset \mathcal{D}$. 

We aim to understand the dimension and structure of the entire degenerate subspace at the product state energy. 
To this end, the typical quantum group construction for the XXZ chain \cite{Kulish1983,Jimbo:1985zk,Drinfeld:1986in,gainutdinov2015CountingSolutionsBethe} is not suitable because of the periodic boundary conditions \cite{PASQUIER1990523}.
Instead, Ref.~\cite{PASQUIER1990523} shows that 
the Hilbert space of the 1D periodic XXZ chain is a module over the aTL algebra, and crucially, 
the periodic Hamiltonian in \refEq{eq:h_1d} is an element of the aTL algebra \cite{PASQUIER1990523,Pinet_2022}. 
Here, we include a brief review \cite{Pinet_2022, PASQUIER1990523, Martin:1993jka,nepomechie2016UniversalBetheAnsatz}.  
We consider a chain of length $N$ with twisted boundary conditions \cite{PASQUIER1990523,Razumov:2001zg}:
\begin{equation}
\label{eq:twisted_bc}
    S_{N+1}^\pm=e^{ i\phi}S_1^\pm, \quad S_{N+1}^z=S_1^z,
\end{equation}
where $e^{i\phi}\equiv w$ quantifies the twist. 
The aTL algebra $\text{aTL}_N(-q-q^{-1})$ is generated by $\Omega_N$, $\Omega_N^{-1}$, $\mathbf{1}$, and $e_i$'s, where $0\leq i\leq N-1$ and $e_{i+N}=e_i$. For $N\geq3$, the generators are related by \cite{temperleylieb1971,PASQUIER1990523,Pinet_2022}:
\begin{alignat}{2}
     &e_ie_i=(-q-q^{-1}) e_i,\quad &&e_ie_{i\pm1}e_i=e_i\;\forall\,1\leq i\leq N,\nonumber\\
     &e_ie_j=e_je_i\; \text{if}\,\lvert i-j\rvert\geq2,\quad &&\Omega_Ne_i=e_{i-1}\Omega_N,\nonumber\\
     &(\Omega_N^{\pm1}e_0)^{N-1}=\Omega_N^{\pm N}(\Omega_N^{\pm1}e_0),\quad &&\Omega_N\Omega_N^{-1}=\Omega_N^{-1}\Omega_N=\mathbf{1}.
     \label{eq:atl_gen}
\end{alignat}
The generators $e_i^\pm$ can be represented with Pauli matrices \cite{Pinet_2022}:
\begin{equation}
\begin{aligned}
    e_i^\pm=&\sigma_i^-\sigma_{i+1}^++\sigma_i^+\sigma_{i+1}^-+(q+q^{-1})\sigma_i^+\sigma_i^-\sigma_{i+1}^+\sigma_{i+1}^-\\
    &-q^{\pm1}\sigma_i^+\sigma_i^--q^{\mp1}\sigma_{i+1}^+\sigma_{i+1}^-\quad \text{for}\,1\leq i\leq N-1,\\
    e_N^\pm=&w^{2}\sigma_N^-\sigma_{1}^++w^{-2}\sigma_1^+\sigma_{N}^-+(q+q^{-1})\sigma_N^+\sigma_N^-\sigma_{1}^+\sigma_{1}^-\\
    &-q^{\pm1}\sigma_N^+\sigma_N^--q^{\mp1}\sigma_{1}^+\sigma_{1}^-.
\end{aligned}
\label{eq:atl_epm}    
\end{equation}
The $\Omega_N$ and $\Omega_N^{-1}$ generators can be written as
\begin{equation}
    \Omega_N=tw^{-\sigma_1^z},\quad \Omega_N^{-1}=t^{-1}w^{\sigma_N^z},
\end{equation}
where $t$ is the translation operator $t\ket{\sigma_1 \sigma_2\dots \sigma_N}=\ket{\sigma_2\dots \sigma_N \sigma_1}$. The Hamiltonian \refEq{eq:h_1d} then becomes \cite{Pinet_2022,PASQUIER1990523}
\begin{equation}
H_\text{XXZ}=\sum_{i=1}^Ne_i^+=\sum_{i=1}^Ne_i^-.
\end{equation}
Note that the module $\mathcal{H}_{N;z}^+$ generated by $\{e_i^+,\Omega_N,\Omega_N^{-1}\}$ and the module $\mathcal{H}_{N;z}^-$ generated by $\{e_i^-,\Omega_N,\Omega_N^{-1}\}$ are two different representations on the spin Hilbert space $\mathcal{H}=\left(\mathbb{C}^2\right)^{\otimes N}$.
The two modules include the same XXZ Hamiltonian \cite{Morin-Duchesne:2012lya}, but they may not be isomorphic \cite{Pinet_2022}. 

The generators $e_i^{\pm}$'s and $\Omega_N$ commute with the total $z$-spin operator $\stot$. Hence, we can decompose $\mathcal{H}_N^\pm$ into a direct sum over different $\stot$ sectors. We introduce $d=-N+2M$, and write
\begin{equation}
\mathcal{H}_{N;w}^\pm=\bigoplus_{M=0}^N\mathcal{H}_{N;d,w}^\pm.
\end{equation}
The spin flip symmetry $d \leftrightarrow -d$ then defines an isomorphism:
\begin{equation}
s:\mathcal{H}_{N;d,w}^\pm\cong\mathcal{H}_{N;-d,w^{-1}}^\mp.
\end{equation}
We can then employ the representation theory of $\text{aTL}_N(-q-q^{-1})$ to understand the spectrum of the 1D periodic XXZ spin chain. Note that the aTL algebra is not the symmetry algebra of the periodic XXZ chain because $H_\text{XXZ}$ does not commute with generic $e_i$'s due to the periodic boundary condition. 

\section{Motivation: Fabricius-McCoy strings}
\label{sec:structure}
Before advancing to the algebraic abstract description, let us
motivate the structure of $\mathcal D$ and origin of the degeneracy when the chain length and the roots of unity are commensurate, i.e., the chain contains an integer number of windings $\gamma N/2\pi$. 
In this case, there is a physically appealing explicit construction by FM strings \cite{fabricius_bethes_2001,fabricius2001_2}, which is built on the conjecture that there exist zero-energy excitations in the form of string operators, which manifests in the BAEs as singular solutions. 
The fundamental idea is that product eigenstates are the starting point for a tower construction of the basis of $\mathcal{D}$: the product eigenstates are boosted by the FM string operators to all possible other basis states.
Note that the towers do not constitute a rigorous proof for the lower bounds for the degeneracy, especially in sectors where $\ell\nmid d$, because the completeness of the FM string operators is--to the best of our knowledge--not proven and hinges on the {string hypothesis} of the Bethe Ansatz \cite{fabricius_bethes_2001,deguchi_xxz_2007,Isler_1993,Hou_2024}. The rigorous construction is delegated to section~\ref{sec:algebra}.
At commensurate roots of unity where $q^N=1$, note that $\ell\mid N$. Refs.~\cite{fabricius_bethes_2001,fabricius2001_2} conjecture and numerically demonstrate that the FM string consists of $\ell$ Bethe roots equally spaced in the imaginary direction around a string center $\alpha^\text{FM}$ \cite{fabricius_bethes_2001,miao_q_2021}: 
\begin{equation}
\label{eq:fm-string}
v_m = \alpha^{\text{FM}}+\frac{2m-1-\ell}{2\ell}i\pi,\quad 1\leq m \leq \ell.
\end{equation}
The special property of FM roots is that they are zero-energy excitations that can be combined with other Bethe solutions \cite{fabricius_bethes_2001,fabricius2001_2,miao_q_2021}. For the tuple of Bethe roots $(v_j)_{j=1}^M$, let FM string solutions be where $j=1,\cdots,\ell$ and other ``ordinary'' solution be where $j=\ell+1,\cdots,M$. The FM roots do not interact with the ordinary roots \cite{fabricius_bethes_2001,miao_q_2021}:
\begin{equation}
    \label{eq:nonscatter}
    \prod_{j=1}^\ell\frac{\sinh{(v_j-v_k+i\gamma)}}{\sinh{(v_j-v_k-i\gamma)}}=1,\quad \ell<k\leq M. 
\end{equation}
Therefore, one can add FM strings to ordinary solutions to evolve to another solution in a different $\stot$ sector \cite{fabricius_bethes_2001,fabricius2001_2,miao_q_2021}. 
The existence of the FM strings as singular solutions of the BAEs at commensurate roots of unity leads to the claim that the conventional method of solving the BAEs is incomplete at roots of unity \cite{fabricius_bethes_2001}. 

Next, we use the structure of FM strings at commensurate roots of unity to conjecture the structure of the degenerate subspace at the product energy level $\mathcal{D}$. 
If $\Delta=0$, there are further symmetries and the system reduces to the XX model, where the Jordan-Wigner transformation enables an analytical solution. We exclude this case from the main text but present the structure of $\mathcal{D}$ in the XX model in Appendix \ref{sec:xx}. 
For $0<\vert\Delta\vert<1$, Refs.~\cite{deguchi_xxz_2007,miao_q_2021} hypothesize that the entire spectrum of the 1D XXZ spin chain at roots of unity can be described using descendant towers, where one adds $\ell$ FM roots to evolve a solution of sector $M$, a {primitive state}, to relate to a degenerate state at sector $M+\ell$ \cite{deguchi_xxz_2007,miao_q_2021}. 
There are two types of zero-energy excitations: the FM strings, and infinity roots of the BAEs \cite{popkov_phantom_2021,miao_q_2021}. The infinity or ``phantom'' roots correspond to solutions of the BAEs where the wavefunction carries no finite rapidity excitations, and hence no nontrivial scattering phases ~\cite{popkov_phantom_2021}.
Thus, there can be several types of descendant towers depending on the initial primitive states, the existence of phantom roots, and the number of corresponding FM roots \cite{miao_q_2021}.

\begin{figure*}[htbp!]
    \centering
    \begin{subfigure}{\textwidth}
    \caption{}
        \centering
        \begin{tikzpicture}[scale=0.9]
          \tikzset{every picture/.style={line width=0.75pt}}
          \tikzset{>=stealth}
          
          \node (a0) at (0,0) {$\emptyset$};
          
          \node (a1) at (-3,-2) {$\{\alpha^{FM}_1\}$};
          \node (a2) at (0,-2) {$\{\alpha^{FM}_2\}$};
          \node (a3) at (3,-2) {$\{\alpha^{FM}_3\}$};
          
          \node (b12) at (-3,-4) {$\{\alpha^{FM}_1, \alpha^{FM}_2\}$};
          \node (b13) at (0,-4) {$\{\alpha^{FM}_1, \alpha^{FM}_3\}$};
          \node (b23) at (3,-4) {$\{\alpha^{FM}_2, \alpha^{FM}_3\}$};
          
          \node (c123) at (0,-6) {$\{\alpha^{FM}_1, \alpha^{FM}_2, \alpha^{FM}_3\}$};
          
          \draw[->] (a0) -- (a1);
          \draw[->] (a0) -- (a2);
          \draw[->] (a0) -- (a3);
          
          \draw[->] (a1) -- (b12);
          \draw[->] (a1) -- (b13);
          \draw[->] (a2) -- (b12);
          \draw[->] (a2) -- (b23);
          \draw[->] (a3) -- (b13);
          \draw[->] (a3) -- (b23);
          
          \draw[->] (b12) -- (c123);
          \draw[->] (b13) -- (c123);
          \draw[->] (b23) -- (c123);
          
          \node (alpha) at (6,-0.66) {$\{\pm\infty\}$};
          \node (alpha1) at (4.5,-2.66) {$\{\pm\infty,\alpha_4^{FM}\}$};
          \node (alpha2) at (7.5,-2.66) {$\{\pm\infty,\alpha_5^{FM}\}$};
          \node (alpha3) at (6,-4.66) {$\{\pm\infty,\alpha_4^{FM},\alpha_5^{FM}\}$};
          
          \draw[->] (alpha) -- (alpha1);
          \draw[->] (alpha) -- (alpha2);
          \draw[->] (alpha1) -- (alpha3);
          \draw[->] (alpha2) -- (alpha3);
          
          \node (beta) at (11,-1.33) {$\{\mp \infty,\mp \infty\}$};
          \node (beta1) at (9.5,-3.33) {$\{\mp \infty,\mp \infty, \alpha_4^{FM}\}$};
          \node (beta2) at (12.5,-3.33) {$\{\mp \infty,\mp \infty, \alpha_5^{FM}\}$};
          \node (beta3) at (11,-5.33) {$\{\mp \infty,\mp \infty,\alpha_4^{FM},\alpha_5^{FM}\}$};
          
          \draw[->] (beta) -- (beta1);
          \draw[->] (beta) -- (beta2);
          \draw[->] (beta1) -- (beta3);
          \draw[->] (beta2) -- (beta3);
          \draw[dashed] (-6, -3) -- (14, -3);
          
          \node at (-5, 0) {$M=0$};
          \node at (-5, -2) {$M=3$};
          \node at (-5, -4) {$M=6$};
          \node at (-5, -6) {$M=9$};
          
          \node at (-5, -0.66) {$M=1$};
          \node at (-5, -1.33) {$M=2$};
          \node at (-5, -2.66) {$M=4$};
          \node at (-5, -3.33) {$M=5$};
          \node at (-5, -4.66) {$M=7$};
          \node at (-5, -5.33) {$M=8$};
        \end{tikzpicture}
        \label{fig:descendant_tower_a}
    \end{subfigure}
    \\[1.5em]
    \begin{subfigure}{\textwidth}
    \caption{}
    \centering
        \begin{tabular}{|c|l|}
        \hline
        Spin sector & Bethe roots\\
        \hline
        $M=0$ & $\emptyset$ \\
        \hline
        $M=1$ &    $\{\infty\},\{-\infty\}$\\
        \hline
        $M=2$ &    $\{\infty,\infty\},\{-\infty,-\infty\}$\\
        \hline
        $M=3$ & $\{0.7365_{FM}\}, \{-0.7365_{FM}\}, \{0_{FM}\}$ \\
        \hline
        $M=4$ & $\{\infty,0.3624_{FM}\},\{\infty,-0.3624_{FM}\},\{-\infty,0.3624_{FM}\},\{-\infty,-0.3624_{FM}\}$\\
        \hline
        $M=5$ & $\{\infty,\infty,0.3624_{FM}\},\{\infty,\infty,-0.3624_{FM}\},\{-\infty,-\infty,0.3624_{FM}\},\{-\infty,-\infty,-0.3624_{FM}\}$\\
        \hline
        $M=6$ & $\{0.7365_{FM},0_{FM}\},\{-0.7365_{FM},0_{FM}\},\{0.7365_{FM},-0.7365_{FM}\}$\\
        \hline
        $M=7$ & $\{\infty,0.3624_{FM},-0.3624_{FM}\},\{-\infty,0.3624_{FM},-0.3624_{FM}\}$\\
        \hline
        $M=8$ & $\{\infty,\infty,0.3624_{FM},-0.3624_{FM}\},\{-\infty,-\infty,0.3624_{FM},-0.3624_{FM}\}$\\
        \hline
        $M=9$ & $\{0.7365_{FM},-0.7365_{FM},0_{FM}\}$\\
        \hline
        \end{tabular}
        \label{fig:descendant_tower_b}
    \end{subfigure}
    \caption{(\subref{fig:descendant_tower_a}) The descendant tower structure of the FM construction reveals a hypothesis for the structure and degeneracy of $\mathcal{D}$. Visualization adapted from Ref.~\cite{miao_q_2021}. For $M=0$, the primitive state is $\ket{\downarrow\dots\downarrow}$, and there are 3 FM strings, resulting in a degeneracy of 8 states. For $M=1,2$, at each sector there are two roots at infinity, and there are 2 FM strings, resulting in a degeneracy of $4\times3=12$ states. In total, there are 24 degenerate states, of which 20 are projected product states and 4 are additional non-product states. The dashed line indicates ``the equator'' where the spin-flip symmetry connects sectors above and below the line, which  becomes manifest in \refFig{fig:diagram_twolines_b} where we include the hidden sectors with non-zero twists. The left tower corresponds to the sequence in Fig.~\ref{fig:diagram_oneline}, while the right two towers are interlinked and correspond to the sequences in Fig.~\ref{fig:diagram_twolines}. (\subref{fig:descendant_tower_b}) Listing of the Bethe roots. $\alpha_{FM}$ indicates 3-strings  of the form $\{\alpha,\alpha-i\pi/3,\alpha+i\pi/3\}$ with string center $\alpha$.}
    \label{fig:descendant_tower}
\end{figure*}

Using the FM construction for commensurate chains, we observe that $\mathcal{D}$ can be decomposed into descendant towers with phantom states as primitive states that are linear combinations of projected eigenstates with $M = 1, \cdots, \ell-1$. 
By combinatorial calculations, we find that for $\Delta\neq0,\pm1$, the number of states collected in the descendant towers is $2^{\frac{N}{\ell}}\ell$ and, assuming linear independence, hypothesize the  degeneracy $g$ at commensurate roots of unity to be
\begin{equation}
    g = 2^{\frac{N}{\ell}}\ell.
    \label{eq:degen-hypo}
\end{equation}
This matches the construction of Ref.~\cite{deguchi2002ConstructionMissingEigenvectors}, which utilizes the FM string construction but is only rigorously proven for sectors where $\stot \equiv 0\mod \ell$. For a single-wounded helix, i.e., $\ell$ = $N$, the degeneracy is $g = 2N$ and only constructed by product eigenstates, recovering \refEq{eq:product_state_degeneracy}. Excess degeneracy appears when helices have multiple windings. Unlike the product state degeneracy, the excess degeneracy increases exponentially with the number of windings and $N$. However, it is still a measure-$0$ set in the thermodynamic limit because for $\ell>1$, $(2^{\frac{N}{\ell}}\ell) / (2^N) = 2^{N(\frac{1}{\ell}-1)}\ell \xrightarrow{N\to \infty} 0$. 

An example of the descendant towers applied to the product state energy for $N=9,\,\Delta=-1/2,\, \ell=3$ is shown in Fig.~\ref{fig:descendant_tower} \cite{miao_q_2021,Hou_2024}. The corresponding Bethe roots are also shown following derivations in Ref.~\cite{miao_q_2021}. According to the FM string construction, we have strings of length of 3: the FM Bethe roots are $\alpha_\text{FM}-i\pi/3,\,\alpha_\text{FM},\,\alpha_\text{FM}+i\pi/3$. From the trivial state at $M=0$ $\ket{\downarrow\dots\downarrow}$, there are 3 different FM strings with string centers $
\alpha_\text{FM}=\pm0.7365,0$, resulting in a degeneracy of 8 states. For $M=1,2$, at each sector there are 2 phantom states with rapidities at infinity, and there are 2 different FM strings with string centers $
\alpha_\text{FM}=\pm0.3624$, resulting in a degeneracy of $4\times3=12$ states. Hence, we expect a total multiplicity of $24$ from \refEq{eq:degen-hypo}. 
We leave the detailed discussion of the calculation of the Bethe roots at roots of unity in Appendix~\ref{sec:transfer_matrix}.

At $\Delta=\pm1$, there are additional symmetry constraints. For commensurate chains (any chain length if $\Delta = 1$ and any even chain length if $\Delta = -1$), each string connects sector $M=j$ to $M=j+1$. Starting from $M=0$, we have $N+1$ states on each string. There is only one string, compared to the $N/\ell$ strings in the anisotropic case. Hence, the degeneracy $g$ at $\Delta=\pm1$ can be hypothesized to be
\begin{equation}
     \label{eq:degen-hypo-pm1}
     g = N+1,
 \end{equation}
matching the numerically obtained values in \refTab{tab:degeneracies}.

\section{The Structure of the Degenerate Subspace}
\label{sec:algebra}

This section offers a rigorous account of the lower bound of the degeneracy hypothesized in \refEq{eq:degen-hypo}, and extends the result to all roots of unity. 
The main result of the article is as follows: 
\begin{theorem}[Lower bound for commensurate case]
    Given a 1D XXZ periodic spin-$1/2$ chain of length $N$ and anisotropy $q$ at a commensurate root of unity ($q^N=1$), let $q^2$ be an $\ell^\text{th}$ primitive root of unity. For $\Delta\neq0,\pm1$, the total degeneracy $g$ at the product state energy satisfies 
\begin{equation}
    \label{eq:degen-thm} 
        g \geq 2^{\frac{N}{\ell}}\ell. 
\end{equation}
\label{thm:degen}
\end{theorem}

\begin{corollary}[Lower bound for incommensurate case]
    At an incommensurate root of unity, let $q^2$ be an $\ell^{\text{th}}$ primitive root of unity. If $N>\ell$, the total degeneracy $g$ at the product state energy at $\Delta\neq0,\pm1$ satisfies:
\begin{equation}
    \label{eq:degen-thm-non}
        g \geq \begin{cases}
            2^{2\lfloor\frac{N}{2\ell}+\frac{1}{2}\rfloor}\quad &\text{for odd } N\text{ and }q^{\ell}=1,\\
            2\quad&\text{for odd } N\text{ and }q^{\ell}=-1,\\
            2^{2\lfloor\frac{N}{2\ell}\rfloor+1}\quad &\text{for even } N.
        \end{cases}
\end{equation}
\label{cor:degen-thm-non}
\end{corollary}

\begin{remark}[Special cases: $\Delta = 0, \pm 1$]%
\label{rem:degen-thm}%
\mbox{}

    \begin{enumerate}
        \item[(i)] If $\Delta = 1$ (XXX model), 
        \begin{equation}
            g = N+1.
            \label{eq:xxx}
        \end{equation}
        \item[(ii)] If $\Delta = -1$ (XX(-X) model), 
        \begin{alignat}{2}
                &g=N+1\quad&&\text{for even } N,\nonumber\\
                &g\geq 2\quad&&\text{for odd } N.
            \label{eq:degen-minusone}
        \end{alignat}
        From the perspective of incommensurate and commensurate chains, for $\Delta=1$, $q$ is a commensurate root of unity for all $N$; for $\Delta =-1$,  $q$ is a commensurate root of unity only for even $N$. 
        
        \item[(iii)] If $\Delta=0$ (XX model), 
        $g$ is "the number of subsets $K$ of $\{1,\dots,N\}$ such that the sum of cosines of the angles in $\{(2j + (1 + (-1)^{|K|})/2 )\pi/N | j \in K\}$ is zero."
        \cite{oeisA392387}, i.e.,
                \begin{align}
        \label{eq:degen-xx}
            g = \operatorname{OEIS_{A392387}}(N).
        \end{align}
    \end{enumerate}
We discuss these special cases in detail in Appendix \ref{sec:xx}. 
\end{remark}
We numerically find that the inequalities are generically saturated, i.e., there is no degeneracy we could not account for, for all roots of unity and chain lengths $N\leq20$, see \refTab{tab:degeneracies}. 

Next, we offer a representation-theoretic explanation on the structure of the degenerate subspace $\mathcal D$ and rigorous proofs of Theorem~\ref{thm:degen} and Corollary~\ref{cor:degen-thm-non}. 

\subsection{Construction of Sequences of Intertwiners}

In this subsection, we use aTL$_N$-linear morphisms to show that there exists a ``hidden'' half-sequence with sectors with twisted boundary conditions in between the sectors with periodic boundary conditions that mediates the degeneracy. 
Ref.~\cite{Pinet_2022} has identified aTL$_N$-linear morphisms between different sectors, the intertwiners, at roots of unity. 
\begin{theorem}[Pinet, Saint-Aubin \cite{Pinet_2022}]
\label{thm:morphism}
    Let $q^2$ be an $\ell$-th primitive root of unity. Suppose $\ell\geq2$. Let $t,d$ be integers, and $|v|=|w|=1$. Define $(t,v)$ to succeed $(d,w)$ if there exists a non-negative integer $m$ satisfying $t=d+2m$ and
\begin{equation}
\label{eq:sucession}
    (a)\; w^2=q^t,\; v=wq^{-m},\quad \text{or}\quad (b)\; w^2=q^{-t},\;v=wq^m.
\end{equation}
    If $(t,v)$ succeeds $(d,w)$, and they correspond to two sectors of the spin chain of length $N$ ($|t|,|d|\leq N$, $t,d\equiv N\mod 2$), the following maps $F^\mp_{(d,v);(t,w)}$ and $E^\pm_{(t,v);(d,w)}$ are aTL$_N$-linear morphisms:
    \begin{align}
\label{eq:morphism}
&F^\mp_{(d,w);(t,v)}:\mathcal{H}_{N;t,v}^\mp\to\mathcal{H}_{N;d,w}^\mp,\\
&E^\pm_{(t,v);(d,w)}:\mathcal{H}_{N;d,w}^\pm\to\mathcal{H}_{N;t,v}^\pm,
\end{align}
where if the succession is through condition (a) of \refEq{eq:sucession}, we take the top signs; if the succession is through condition (b) of \refEq{eq:sucession}, we take the bottom signs. 
In particular, if $(d,w)$ is the direct successor of $(t,v)$ through condition (a) (resp. (b)), $(s,u)$ is the successor of $(d,w)$ through condition (a) (resp. (b)), and $\frac{1}{2}(s-d)\not\equiv 0\mod{\ell}$ (resp. $\frac{1}{2}(s+d)\not\equiv 0\mod{\ell}$), then the following sequence is exact:
\begin{equation}
    \mathcal H^\pm_{N;t,v}\xrightarrow{F^\pm} \mathcal H^\pm_{N;d,w}\xrightarrow{F^\pm} \mathcal H^\pm_{N;s,u}.
\end{equation}
\end{theorem}
The FM strings are a realization of these intertwiners, as remarked by Ref~\cite{Pinet_2022}. 
Here, one can understand $F$'s to be the ``spin-lowering'' strings that connect sector $d$ with sector $d-2\ell$ in the FM construction, while $E$'s are the ``spin-raising'' strings that connect sector $d$ with sector $d+2\ell$. We give explicit expressions for the morphisms which are actions of Lusztig's divided powers~\cite{LUSZTIG1988} and demonstrate this connection in Appendix~\ref{appendix:connection_FM_aTL}. 

Next, we use Theorem~\ref{thm:morphism} to derive the structure of the degenerate subspace at the product state energy. 
First, we prove the following lemma regarding the spectrum of $H_\text{XXZ}$ that enable us to transport eigenvalues between different sectors. The detailed proof is shown in Appendix~\ref{sec:proof}. 
\begin{lemma}[Transport of generalized eigenvalues]
    If an aTL$_n$-linear morphism $\mu:\mathcal{H}_{N;t,v} \to \mathcal{H}_{N;d,w}$ exists, consider $H_\text{XXZ}$ as an element of the aTL$_N$ algebra. Then $\mu$ preserves energy eigenvalues: any energy eigenvector in $\mathcal{H}_{N;t,v}$ is mapped either to zero or to an energy eigenvector in $\mathcal{H}_{N;d,w}$ with the same energy. 
    \label{lem:transport}
\end{lemma}

Through the intertwiner relationships in Theorem \ref{thm:morphism}, we can transport eigenvalues between different sectors even if the representation of $H_\text{XXZ}$ differ and the twists of the two chains are different---thus making the Hamiltonian take different forms on the two chains. As long as the sectors are connected via \refEq{eq:sucession}, they possess eigenstates with the same energy. 

Note that for the fully polarized states, the energy is independent of the twist parameter $w$. We can therefore start with the fully polarized states with any twist and construct sequences of sectors containing states with the desired product state energy. Consequently, each sequence of intertwiner-connected sectors determines a set of degenerate states across those sectors.

Let us look at $N=9$, $\Delta=-1/2$, $q=e^{\frac{2}{3}\pi i}$, $\ell=3$ as a comprehensive example and construct the sequences of sectors connected by intertwiners $F^\pm$. We denote each Hilbert space $\mathcal{H}_{N;t,v}^\pm$ as $(N;t,v)^\pm$ for ease of reading. Fig.~\ref{fig:diagram_oneline} shows the sequence of morphisms that starts with the fully polarized sector with zero twist. This is an exact sequence by Theorem~\ref{thm:morphism}, and the state in the fully polarized sector $(9;9,1)^+$ is a highest weight vector in the $\mathfrak{sl}_2$ algebra \cite{deguchi_xxz_2007}. This sequence of morphisms induces a degeneracy of $2^{N/\ell}=8$ from the fully polarized sector. This is analogous to the tower constructed and proven in Ref.~\cite{deguchi_xxz_2007} for sectors where $\ell\mid d$. 

Fig.~\ref{fig:diagram_twolines} shows the sequences of morphisms that start with fully polarized sectors with non-zero twists that are powers of $q$. We can observe that compared to Fig.~\ref{fig:diagram_oneline}, there now exists hidden sectors with non-zero twists and $\ell\mid d$ (labeled in red) that connect and mediate the degeneracy in the periodic sectors where $\ell\nmid d$ (labeled in black). 

Fig.~\ref{fig:diagram_twolines_b} visualizes the same intertwiner structure as Fig.~\ref{fig:diagram_twolines} in the tower representation, revealing how the right two descendant towers from Fig.~\ref{fig:descendant_tower_a} are in fact interconnected through hidden sectors with nonzero twists. The blue solid arrows represent the $F^-$ intertwiners (corresponding to horizontal arrows in Fig.~\ref{fig:diagram_twolines}) that connect sectors within the $\mathcal{H}^-$ representation spaces, while the cyan dotted arrows represent the $F^+$ intertwiners (corresponding to diagonal arrows in Fig.~\ref{fig:diagram_twolines}) that connect sectors within the $\mathcal{H}^+$ representation spaces. The tower visualization with the hidden sectors make the spin-flip symmetry manifest, and shows explicitly that what appears as two separate descendant towers in the FM string picture is actually a single interconnected structure unified by the representation theory of the aTL algebra.  
\begin{figure*}[htbp!]
    \centering
    \begin{subfigure}{\textwidth}
            \caption{}
        \vspace*{-2\baselineskip}
        \centering
        \begin{tikzcd}
          (9;9,1)^+
            \arrow[r]
          &
          (9;3,1)^+
            \arrow[r]
          &
          (9;-3,1)^+
            \arrow[r]
          &
          (9;-9,1)^+
        \end{tikzcd}
        \label{fig:diagram_oneline}
    \end{subfigure}
    \\[0.5em]
    \begin{subfigure}{\textwidth}
            \caption{}
        \centering
        \begin{tikzcd}
          |[draw=cbRed]| \color{cbRed}{(9;9,q)^\mp}
            \arrow[dr]
            \arrow[r]
          &
          (9;5,1)^\mp
            \arrow[dr]
             \arrow[r]
          &
          |[draw=cbRed]| \color{cbRed}{(9;3,q)^\mp}
            \arrow[dr]
             \arrow[r]
          &
          (9;1,1)^\mp
            \arrow[dr]
             \arrow[r]
          &
          |[draw=cbRed]| \color{cbRed}{(9;-3,q)^\mp}
            \arrow[dr]
             \arrow[r]
          &
          (9;-7,1)^\mp
            \arrow[dr]
             \arrow[r]
          &
          |[draw=cbRed]| \color{cbRed}{(9;-9,q)^\mp}
          \\
          |[draw=cbRed]| \color{cbRed}{(9;9,-q)^\mp}
            \arrow[ur]
            \arrow[r]
          &
          (9;7,1)^\mp
          \arrow{ur}
          \arrow[r]
          &
          |[draw=cbRed]| \color{cbRed}{(9;3,-q)^\mp}
          \arrow{ur}
          \arrow[r]
          &
          (9;1,1)^\mp
          \arrow{ur}
          \arrow[r]
          &
          |[draw=cbRed]| \color{cbRed}{(9;-3,-q)^\mp}
          \arrow{ur}
          \arrow[r]
          &
          (9;-5,1)^\mp
          \arrow{ur}
          \arrow[r]
          &
          |[draw=cbRed]| \color{cbRed}{(9;-9,-q)^\mp}
        \end{tikzcd}
        \label{fig:diagram_twolines}
    \end{subfigure}
    \\[0.5em]
    \begin{subfigure}{\textwidth}
            \caption{}
        \vspace*{-1\baselineskip}
        \centering
        \begin{tikzpicture}[scale=0.9]
          \tikzset{every picture/.style={line width=0.75pt}}
          \tikzset{>=stealth}
          \usetikzlibrary{fit, decorations.pathmorphing, arrows.meta}
          
            \node (alpha)[draw, minimum height=1cm, align=center] at (0,-0.8712) {$\{\pm\infty\}$ \\ \scriptsize 2 states};
            \node (alpham)[draw, minimum height=1cm, align=center] at (0,-3.5112) {$\{\pm\infty,\alpha_4^{FM}\}\quad\{\pm\infty,\alpha_5^{FM}\}$ \\ \scriptsize 4 states};
            \node (alpha3)[draw, minimum height=1cm, align=center] at (0,-6.1512) {$\{\pm\infty,\alpha_4^{FM},\alpha_5^{FM}\}$ \\ \scriptsize 2 states};

            \draw[dashed, thick] (alpha) -- (alpham);
            \draw[dashed, thick] (alpham) -- (alpha3);
            
            \node (h0) [twistColor] at (4, 0.8712) {twist $q$};
            \node (h1) [statesNode] at (4,0) {$1$ state};
            \node (h2) [statesNode, minimum width=1.5cm] at (4,-2.64) {$3$ states};
            \node (h3) [statesNode, minimum width=1.5cm] at (4,-5.28) {$3$ states};
            \node (h4) [statesNode] at (4,-7.92) {$1$ state};
            
            \draw[connection] (h1) -- (h2); 
            \draw[connection] (h2) -- (h3); 
            \draw[connection] (h3) -- (h4);
            
            \node (g0) [twistColor] at (7, 0.8712) {twist $-q$};
            \node (g1) [statesNode] at (7,0) {$1$ state};
            \node (g2) [statesNode, minimum width=1.5cm] at (7,-2.64) {$3$ states};
            \node (g3) [statesNode, minimum width=1.5cm] at (7,-5.28) {$3$ states};
            \node (g4) [statesNode] at (7,-7.92) {$1$ state};
            
            \draw[connection] (g1) -- (g2); 
            \draw[connection] (g2) -- (g3); 
            \draw[connection] (g3) -- (g4);
            
            \node (beta)[draw, minimum height=1cm, align=center] at (11,-1.7556) {$\{\mp \infty,\mp \infty\}$ \\ \scriptsize 2 states};
            \node (betam)[draw, minimum height=1cm, align=center] at (11,-4.3956) {$\{\mp \infty,\mp \infty, \alpha_4^{FM}\}\quad\{\mp \infty,\mp \infty, \alpha_5^{FM}\}$ \\ \scriptsize 4 states};
            \node (beta3)[draw, minimum height=1cm, align=center] at (11,-7.0356) {$\{\mp \infty,\mp \infty,\alpha_4^{FM},\alpha_5^{FM}\}$ \\ \scriptsize 2 states};
          
          \draw[dashed, thick] (beta) -- (betam);
          \draw[dashed, thick] (betam) -- (beta3);

            \node at (-4, 0) {$M=0$};
            \node at (-4, -0.8712) {$M=1$};
            \node at (-4, -1.7556) {$M=2$};
            \node at (-4, -2.64) {$M=3$};
            \node at (-4, -3.5112) {$M=4$};
            \node at (-4, -4.3956) {$M=5$};
            \node at (-4, -5.28) {$M=6$};
            \node at (-4, -6.1512) {$M=7$};
            \node at (-4, -7.0356) {$M=8$};
            \node at (-4, -7.92) {$M=9$};

          \draw[dashed] (-5, -3.96) -- (14, -3.96);

          \draw[connectionSame] (alpha.north) to (h1.west);
          \draw[connectionSame] (h2.west) to (alpha.south);
          \draw[connectionSame] (alpham.north) to (h2.west);
          \draw[connectionSame] (h3.west) to (alpham.south);
          \draw[connectionSame] (alpha3.north) to (h3.west);
          \draw[connectionSame] (h4.west) to (alpha3.south);

          \draw[connectionSame] (beta.north) to (g1.east);
          \draw[connectionSame] (g2.east) to (beta.south);
          \draw[connectionSame] (betam.north) to (g2.east);
          \draw[connectionSame] (g3.east) to (betam.south);
          \draw[connectionSame] (beta3.north) to (g3.east);
          \draw[connectionSame] (g4.east) to (beta3.south);

          \draw[connectionDifferent] (alpha.east) to (g1.south);
          \draw[connectionDifferent] (g2.north) to (alpha.east);
          \draw[connectionDifferent] (alpham.east) to (g2.south);
          \draw[connectionDifferent] (g3.north) to (alpham.east);
          \draw[connectionDifferent] (alpha3.east) to (g3.south);
          \draw[connectionDifferent] (g4.north) to (alpha3.east);

          \draw[connectionDifferent] (beta.west) to (h1.south);
          \draw[connectionDifferent] (h2.north) to (beta.west);
          \draw[connectionDifferent] (betam.west) to (h2.south);
          \draw[connectionDifferent] (h3.north) to (betam.west);
          \draw[connectionDifferent] (beta3.west) to (h3.south);
          \draw[connectionDifferent] (h4.north) to (beta3.west);
          
\end{tikzpicture}
        \label{fig:diagram_twolines_b}
    \end{subfigure}
    
    \caption{aTL intertwiners reveal how sectors with different $\stot$ and twists are connected by intertwiners $F^\pm$ starting from the fully polarized sectors with integer power twists of $q$, relating to the FM construction in \refFig{fig:descendant_tower} with $N=9$, $q=e^{\frac{2}{3}\pi i}$, $\ell=3$. %
    For brevity, we denote $\mathcal{H}_{N;t,v}^\pm \equiv (N;t,v)^\pm$. Reversing the $F^\pm$ arrows yields $E^\mp$. (\subref{fig:diagram_oneline}) The exact sequence of sectors where $\ell\mid d$, connected by $F^+$. The state in the fully polarized sector $(9;9,1)^+$ is a highest weight vector in the $L(\mathfrak{sl}_2)$ algebra \cite{deguchi_xxz_2007}, inducing a degeneracy of $2^{N/\ell}=8$.  This corresponds to the left tower in Fig.~\ref{fig:descendant_tower_a}. (\subref{fig:diagram_twolines}) Sequences of sectors that start with non-zero twists. Horizontal arrows indicate the $F^-$ intertwiners between $\mathcal H^-$ spaces, and diagonal arrows indicate the $F^+$ intertwiners between $\mathcal H^+$ spaces. The boxed red sectors highlight the ``hidden'' sectors with non-zero twists and $\ell\mid d$ which explain the total degeneracy of the zero twist sectors (black). For each horizontal sequence, the red sectors have a total degeneracy of $2^{N/\ell}=8$, 
    resulting in a total degeneracy across all periodic sectors of $16$. 
    (\subref{fig:diagram_twolines_b}) The aTL algebra perspective reveals that the right two descendant towers in Fig.~\ref{fig:descendant_tower_a} are interconnected via hidden towers with nonzero twists $q$ and twist $-q$ (red). Blue solid arrows represent $F^-$ intertwiners (horizontal arrows in~\subref{fig:diagram_twolines}); cyan dotted arrows represent $F^+$ intertwiners (diagonal arrows in~\subref{fig:diagram_twolines}). The hidden towers account for the degeneracies of $8$ per periodic tower. With the hidden twisted sectors included, the spin-flip symmetry become manifest: the four-tower structure is symmetric about the equator (dashed line at $M=4.5$), with each sector at $M$ having a mirror partner at $9-M$ connected by the same intertwiner pattern.}
    \label{fig:diagram_combined}
\end{figure*}
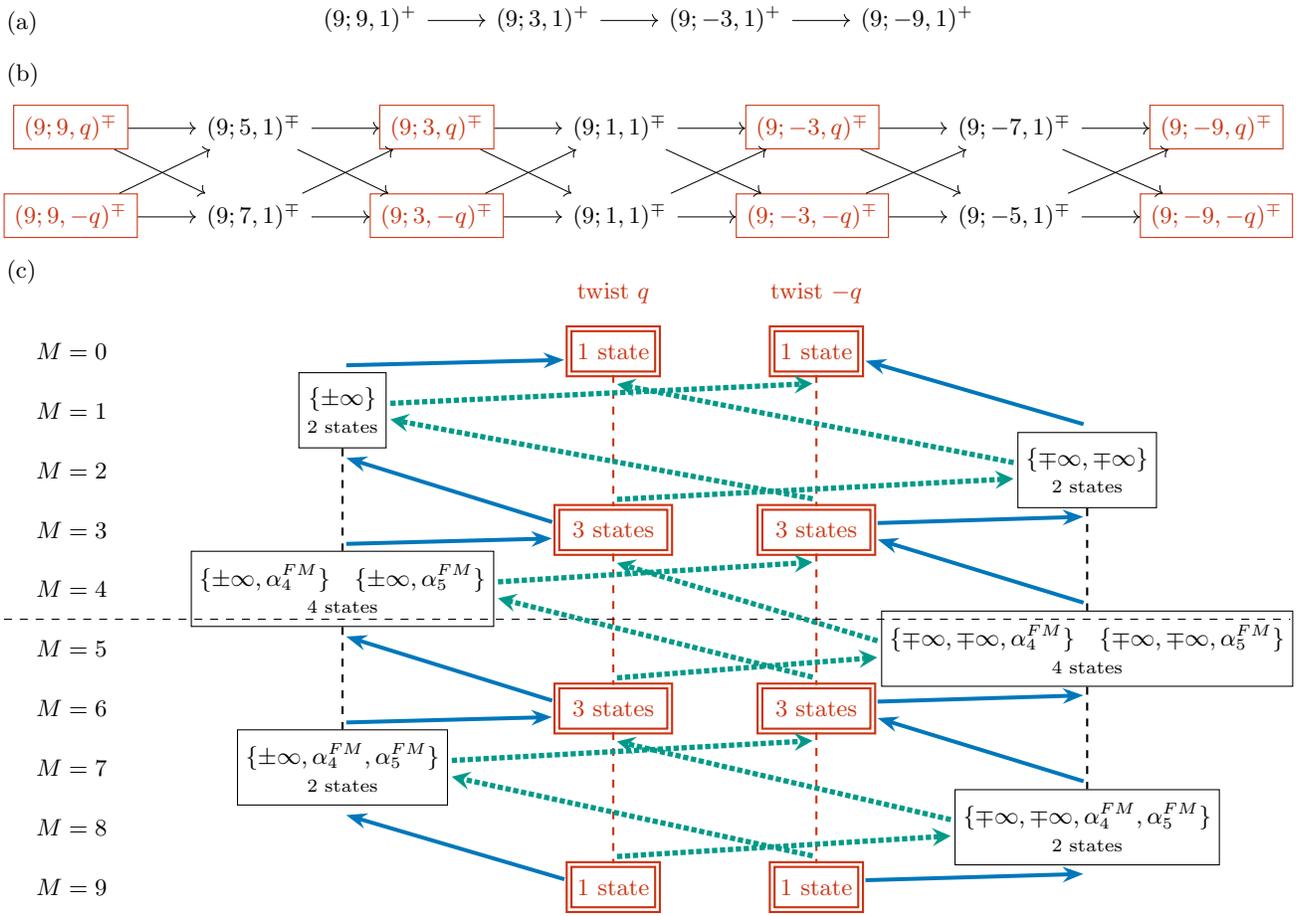

In general, since $q^2$ is an $\ell^{\text{th}}$ primitive root of unity, we have $q^\ell=\pm1$. 
If $q^\ell=1$, $\ell$ is odd (else we can find a smaller integer $\ell/2$ that satisfies the root-of-unity condition of $q^{2\ell}=1$), and all integer powers of $q$ can be expressed as $q^{\pm r}$ with $0\leq r< \ell/2$. The intertwiner sequences (Fig.~\ref{fig:diagram_oddl}) generalize Figs.~\ref{fig:diagram_combined} and take two forms: (i) One sequence starting from the fully polarized sector $(N;N,1)^+$ connects only periodic boundary sectors (Fig.~\ref{fig:diagram_oddl_a}); (ii) For each $1\leq r<\ell/2$, two connected sequences starting from $(N;N,q^{\pm r})^\mp$ connect periodic sectors through hidden twisted boundary sectors (red in Fig.~\ref{fig:diagram_oddl_b}), giving $\ell-1$ such sequences. The hidden twisted sectors are crucial for mediating degeneracies in periodic sectors where $\ell \nmid d$.

\begin{figure*}[htbp!]
    \centering
    \begin{subfigure}{\textwidth}
            \caption{}
        \centering
        \vspace* {-2em}
        \begin{tikzcd}
          \cdots
            \arrow[r]
          &
          (N;k\ell,1)^+
            \arrow[r]
          &
          (N;(k-2)\ell,1)^+
            \arrow[r]
          &
          \cdots
            \arrow[r]
          &
          (N;-k\ell,1)^+
            \arrow[r]
          &
          \cdots
        \end{tikzcd}
        \label{fig:diagram_oddl_a}
    \end{subfigure}
    \\[1.5em]
    \begin{subfigure}{\textwidth}
            \caption{}
        \centering
        \begin{tikzcd}
          \cdots
            \arrow[dr]
            \arrow[r]
          &
          |[draw=cbRed]| \color{cbRed}{(N;k\ell,q^{-r})^\mp}
            \arrow[dr]
             \arrow[r]
          &
          (N;(k-2)\ell+2r,1)^\mp
            \arrow[dr]
             \arrow[r]
          &
          |[draw=cbRed]| \color{cbRed}{(N;(k-2)\ell,q^{-r})^\mp}
            \arrow[dr]
             \arrow[r]
          &
          \cdots
            \arrow[dr]
             \arrow[r]
          &
          |[draw=cbRed]| \color{cbRed}{(N;-k\ell,q^{-r})^\mp}
            \arrow[dr]
             \arrow[r]
          &
          \cdots
          \\
         \cdots
            \arrow[ur]
            \arrow[r]
          &
          |[draw=cbRed]| \color{cbRed}{(N;k\ell,q^{r})^\mp}
          \arrow{ur}
          \arrow[r]
          &
          (N;k\ell-2r,1)^\mp
          \arrow{ur}
          \arrow[r]
          &
          |[draw=cbRed]| \color{cbRed}{(N;(k-2)\ell,q^{r})^\mp}
          \arrow{ur}
          \arrow[r]
          &
          \cdots
          \arrow{ur}
          \arrow[r]
          &
          |[draw=cbRed]| \color{cbRed}{(N;-k\ell,q^{r})^\mp}
          \arrow{ur}
          \arrow[r]
          &
          \cdots
        \end{tikzcd}
        \label{fig:diagram_oddl_b}
    \end{subfigure}
    \caption{The general sectors connected by intertwiners $F^{\pm}$ for $q^\ell=1$, $0<r\leq\ell/2$, and $r\in\mathbb{N}$. For different $N$, as long as they admit these sectors, they are part of the sequences that truncate at the appropriate sectors. In general, all odd $N$ are in one set of sequences and all even $N$ are in another set of sequences for each $\ell$. (\subref{fig:diagram_oddl_a}) The sequence that start with the periodic boundary sector ($w=1$). (\subref{fig:diagram_oddl_b}) Sequences starting with non-zero twists, which include hidden half-sequences of twisted sectors (boxed in red, $w=q^{\pm r}$) that mediate connections between periodic boundary sectors.}
    \label{fig:diagram_oddl}
\end{figure*}
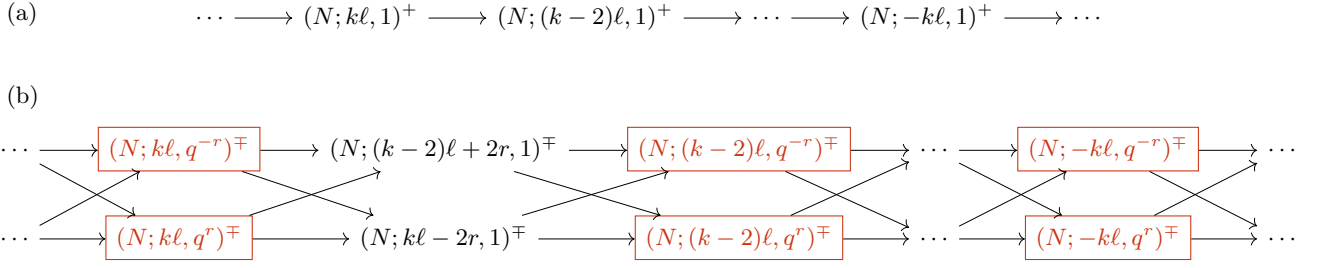

If $q^\ell=-1$, the structure differs for even versus odd $N$. For even $N$, only sectors with $d=2k\ell$ ($k\in\mathbb{Z}$) appear. The sequences (Fig.~\ref{fig:diagram_ql-1_evenN}) are similar to the $q^\ell=1$ case but with alternating twist signs. Pure periodic sequences alternate between $w=1$ and $w=-1$ (Fig.~\ref{fig:diagram_ql-1_evenN_a},~\ref{fig:diagram_ql-1_evenN_b}), while twisted-mediated sequences use hidden sectors with twists $w=\pm q^{\pm r}$ (Fig.~\ref{fig:diagram_ql-1_evenN_c}, red sectors). Blue highlighting marks sectors whose twist has flipped sign relative to Fig.~\ref{fig:diagram_oddl}.
For odd $N$, in the commensurate case where $\ell\mid N$, $\ell$ has to be odd. We can prove that the generic sector $(-(2k+1)\ell+2d,1)$ where $k,d\in\mathbb{Z}$ and $0\leq d<\ell$ does not have a successor:
\begin{lemma}[Odd $N$ and $q^{\ell}=-1$] 
    Let $q^2$ be an $\ell^{\text{th}}$ primitive root of unity. If $q^\ell=-1$, $k,d\in\mathbb{Z}$, $\ell$ even, $0\leq d<\ell$, then $(-(2k+1)\ell+2d,1)$ does not have a successor as defined by \refEq{eq:sucession}. 
    \label{lem:oddl_oddN}
\end{lemma}
\begin{proof}
    By definition, a successor of $(d_0,w_0)$ is obtained by finding the smallest integer $s>d_0$ such that 
    \begin{equation}
         q^{\pm s} = w_0^2,
    \end{equation}
    and then setting $k' = (s-d_0)/2$, $d' = d_0+2k'$, and $w' = w_0 q^{k'}$. 
    In our case we have
    \begin{equation}
          d_0 = -(2k+1)\ell + 2d,\qquad w_0=1,
    \end{equation}
    so the condition becomes $q^{s}=1$. This forces $s$ to be a multiple of $2\ell$, say $s=2\ell m$ with $m\in\mathbb{Z}$. Thus
    \begin{equation}
        k'=\frac{s-d_0}{2}=\frac{2\ell m+ (2k+1)\ell - 2d}{2} = \ell(m+k)+\frac{\ell}{2}-d.
    \end{equation}
    If $\ell$ is odd, $k'$ is not an integer. 
\end{proof}

\begin{figure*}[htbp!]
    \centering
    \begin{subfigure}{\textwidth}
            \caption{}
        \centering
        \vspace*{-2em}
        \begin{tikzcd}
          \cdots
            \arrow[r]
          &
          (N;2k\ell,1)^+
            \arrow[r]
          &
          |[draw= cbBlue, dashed]| \color{cbBlue}{(N;2(k-1)\ell,-1)^+}
            \arrow[r]
          &
          \cdots
            \arrow[r]
          &
          (N;-2k\ell,1)^+
            \arrow[r]
          &
          \cdots
        \end{tikzcd}
        \label{fig:diagram_ql-1_evenN_a}
    \end{subfigure}
    \\[1.5em]
    \begin{subfigure}{\textwidth}
            \caption{}
        \centering
        \vspace*{-2em}
        \begin{tikzcd}
          \cdots
            \arrow[r]
          &
          |[draw= cbBlue, dashed]| \color{cbBlue}{(N;2k\ell,-1)^+}
            \arrow[r]
          &
          (N;2(k-1)\ell,1)^+
            \arrow[r]
          &
          \cdots
            \arrow[r]
          &
          |[draw= cbBlue, dashed]| \color{cbBlue}{(N;-2k\ell,-1)^+}
            \arrow[r]
          &
          \cdots
        \end{tikzcd}
        \label{fig:diagram_ql-1_evenN_b}
    \end{subfigure}
    \\[1.5em]
    \begin{subfigure}{\textwidth}
            \caption{}
        \centering
        \begin{tikzcd}
          \cdots
            \arrow[dr]
            \arrow[r]
          &
          |[draw=cbRed]| \color{cbRed}{(N;2k\ell,q^{-r})^\mp}
            \arrow[dr]
             \arrow[r]
          &
          |[draw= cbBlue, dashed]| \color{cbBlue}{(N;2(k-1)\ell+2r,-1)^\mp}
            \arrow[dr]
             \arrow[r]
          &
          |[draw=cbRed]| \color{cbRed}{(N;2(k-1)\ell,-q^{-r})^\mp}
            \arrow[dr]
             \arrow[r]
          &
            |[draw=cbBlue, dashed]| \color{cbBlue}{(N;2(k-2)\ell+2r,-1)^\mp}
            \arrow[dr]
             \arrow[r]
          &
          \cdots
          \\
         \cdots
            \arrow[ur]
            \arrow[r]
          &
          |[draw=cbRed]| \color{cbRed}{(N;2k\ell,q^{r})^\mp}
          \arrow{ur}
          \arrow[r]
          &
          (N;2k\ell-2r,1)^\mp
          \arrow{ur}
          \arrow[r]
          &
          |[draw=cbRed]| \color{cbRed}{(N;2(k-1)\ell,-q^{r})^\mp}
          \arrow{ur}
          \arrow[r]
          &
            {(N;2(k-1)\ell-2r,1)^\mp}
          \arrow{ur}
          \arrow[r]
          &
          \cdots
        \end{tikzcd}
        \label{fig:diagram_ql-1_evenN_c}
    \end{subfigure}
    \caption{The general sectors connected by intertwiners $F^{\pm}$ for $q^\ell=-1$, even $N$, $0<r\leq\ell/2$, and $r\in\mathbb{N}$. Solid red boxes indicate hidden twisted boundary sectors ($w=\pm q^{\pm r}$). Dashed blue boxes indicate sectors whose twists acquire a factor of $-1$ compared to the $q^\ell=1$ case (Fig.~\ref{fig:diagram_oddl}). (\subref{fig:diagram_ql-1_evenN_a}) The sequence that starts with periodic boundary sector ($w=1$). Some sectors with $w=-1$ appear. (\subref{fig:diagram_ql-1_evenN_b}) The sequence that starts with anti-periodic boundary sector ($w=-1$). This parallels~\subref{fig:diagram_ql-1_evenN_a} with all twists multiplied by $-1$. (\subref{fig:diagram_ql-1_evenN_c}) Sequences starting with non-zero twists that include hidden half-sequences of twisted sectors ($w=q^{\pm r}$).}
    \label{fig:diagram_ql-1_evenN}
\end{figure*}
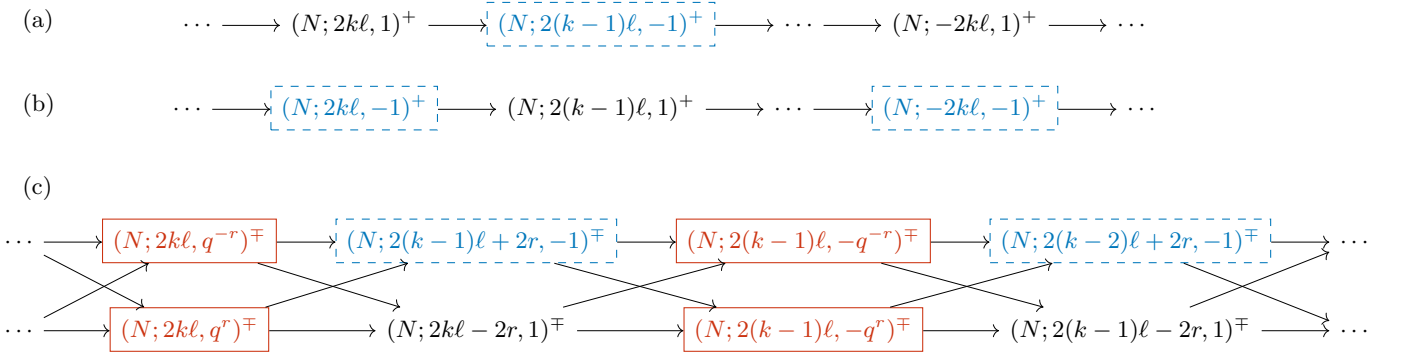

\subsection{Dimension of the Eigenspaces}
As a main result, we derive the dimension of the images of the morphism restricted to the eigenspaces of $H_\text{XXZ}$ at the product state energy eigenvalue. 

When the twist takes the discrete values $\omega = e^{i\phi} = q^{2p}$ for $p \in \tfrac{1}{2}\mathbb{Z}$, an enlarged symmetry emerges in magnetization sectors where $\ell \mid d$ \cite{deguchi_xxz_2007, deguchi2002ConstructionMissingEigenvectors}. For these sectors, the transfer matrix and XXZ Hamiltonian commute with a family of nonlocal operators that generate the Borel subalgebra of the $\mathfrak{sl}_2$ loop algebra $L(\mathfrak{sl}_2)$~\cite{Deguchi03,Korff03,deguchi_xxz_2007}.

The key structural result is that every regular Bethe ansatz eigenvector, meaning one constructed from finite, distinct Bethe roots (in our case the fully polarized states), is a highest-weight vector of $L(\mathfrak{sl}_2)$ in these sectors \cite{deguchi_xxz_2007}. Highest-weight representations of the loop algebra are characterized by a Drinfeld polynomial \cite{Deguchi03,deguchi_xxz_2007}
\begin{equation}
    P(u) = \prod_{j=1}^{r}(u - a_j)
\end{equation}
of degree $r$, where the $a_j$ are called evaluation parameters. When these evaluation parameters are nonzero and pairwise distinct, the corresponding highest-weight representation is irreducible with dimension $2^r$ \cite{deguchi_xxz_2007}.

Consequently, the symmetry multiplet built on a highest-weight Bethe state has multiplicity $2^r$, where $r$ is the degree of its Drinfeld polynomial. For the fully polarized state $\ket{\uparrow\cdots\uparrow}$ in an $N$-site chain with $\ell \mid N$ and discrete twist $q^{2p}$, the Drinfeld polynomial has degree $r = N/\ell$, yielding a degeneracy of $2^{N/\ell}$ \cite{deguchi_xxz_2007,miao_q_2021,deguchi2002ConstructionMissingEigenvectors}. 
In the case of Fig.~\ref{fig:diagram_oddl_b},~\ref{fig:diagram_ql-1_evenN_c}, note that the horizontal sequences are exact by Theorem~\ref{thm:morphism}. We use that exactness is preserved upon restriction to a fixed $H$-eigenspace. The detailed proof is in Appendix~\ref{sec:proof}. 

\begin{lemma}[Restriction of exact sequences]\label{lem:restricted-exact}
Let
\begin{equation}
0 \longrightarrow V_0 \xrightarrow{d_0} V_1 \xrightarrow{d_1} V_2 \xrightarrow{d_2}
\cdots \xrightarrow{d_{n-1}} V_n \longrightarrow 0
\end{equation}
be an exact sequence of finite-dimensional complex vector spaces. Suppose a group $G$ acts on each $V_i$ and all maps $d_i$ are $G$-linear intertwiners. Fix $H\in G$ and denote by $\rho_{V_i}(H)$ the representation of $H$ on $V_i$. Assume each $\rho_{V_i}(H)$ is diagonalizable. For a fixed eigenvalue $h\in\mathbb{C}$ write $V_{i,h}\equiv\ker\big(\rho_{V_i}(H)-hI\big)$
for the $h$-eigenspace in $V_i$, and let $d_{i,h}:V_{i,h}\to V_{i+1,h}$ be the restriction of $d_i$.
Then the restricted sequence
\begin{equation}
0 \longrightarrow V_{0,h} \xrightarrow{d_{0,h}} V_{1,h} \xrightarrow{d_{1,h}}
\cdots \xrightarrow{d_{n-1,h}} V_{n,h} \longrightarrow 0
\end{equation}
is exact.
\label{lemma:restriction}
\end{lemma}

Using the fact that exact sequences have vanishing Euler characteristics \cite{Weibel_1994,nakahara2003geometry}, we can arrive at the following corollary:
\begin{corollary}[Dimension of restricted exact sequences]
Under the conditions of Lemma \ref{lem:restricted-exact}, for each eigenvalue $h$, the alternating sum of dimensions vanishes:
\begin{equation}
\sum_{j=0}^n (-1)^j \dim V_{j,h} = 0.
\end{equation}
\label{cor:alternate}
\end{corollary}

For $q^\ell=1$, the periodic and twisted sectors alternate in the sequences in Fig.~\ref{fig:diagram_oddl_b}. From Corollary~\ref{cor:alternate}, the degeneracy of the periodic sectors therefore equals the degeneracy of the twisted sectors where $\ell\mid d$. For the twisted sectors, we know the degeneracy because of the $L\left(\mathfrak{sl}_2\right)$ symmetry, with the highest weight vector being the fully-polarized state with degree $N/\ell$. Hence, the periodic sectors on the sequence also have a total degeneracy of $2^{N/\ell}$. We have $\ell-1$ such sequences with alternating periodic and twisted sectors. In addition, we have one sequence with only the periodic sectors (Fig.~\ref{fig:diagram_oddl_a}), with degeneracy of $2^{N/\ell}$. We thus arrive at Theorem~\ref{thm:degen}. For $q^\ell=-1$, the argument is analogous except for having to keep track of the periodic and anti-periodic sectors. 

\subsection{Degeneracies at Incommensurate Roots of Unity}

Beyond commensurate roots of unity, we observe enhanced degeneracies at product state energy at incommensurate roots of unity $q^N \neq 1$, see \refTab{tab:degeneracies}. We can employ the same intertwiner construction to explain this. For instance, if we look at $\ell=3,N=7$, the intertwiner sequences are truncated from the case of $\ell=3,N=9$ (Fig.~\ref{fig:diagram_twolines}), which we show in Fig. \ref{fig:diagam_non_n_div}. 
We can see that the hidden half-sequence of non-zero twist sectors again determines the multiplicity of the energy level at each sector: here, the red sectors' multiplicities are governed by the $L(\mathfrak{sl}_2)$ algebra. Each of the two sequences of $\mathcal{H}_{7,3,q}^\mp\to \mathcal{H}_{7,-3,q}^\mp$ and $\mathcal{H}_{7,3,-q}^\mp\to \mathcal{H}_{7,-3,-q}^\mp$ has a multiplicity of $2$ and each sector then has multiplicity of $1$. This then translate to the multiplicities of the periodic sectors, and the total multiplicity is $4$. The difference between incommensurate chains and commensurate chains is that this sequence is the only sequence containing sectors of zero twists. There is no analog to Fig. \ref{fig:diagram_oneline}, since all sectors on those chains obey $\ell\mid d$, and hence do not start with sectors where $d=N$ if $\ell\nmid N$, which is where the product state energy comes from. 

\begin{figure*}[htbp!]
    \centering
    \begin{tikzcd}
  &
  |[draw=cbRed]| \color{cbRed}{\mathcal{H}_{7,3,q}^\mp}
    \arrow[dr]
     \arrow[r]
  &
  \mathcal{H}_{7,1,1}^\mp
    \arrow[dr]
     \arrow[r]
  &
  |[draw=cbRed]| \color{cbRed}{\mathcal{H}_{7,-3,q}^\mp}
     \arrow[r]
  &
  \mathcal{H}_{7,-7,1}^\mp
  \\
  \mathcal{H}_{7,7,1}^\mp
  \arrow{ur}
  \arrow[r]
  &
  |[draw=cbRed]| \color{cbRed}{\mathcal{H}_{7,3,-q}^\mp}
  \arrow{ur}
  \arrow[r]
  &
  \mathcal{H}_{7,1,1}^\mp
  \arrow{ur}
  \arrow[r]
  &
  |[draw=cbRed]| \color{cbRed}{\mathcal{H}_{7,-3,-q}^\mp}
  \arrow{ur}
  &
\end{tikzcd}
    \caption{As an example of the incommensurate case, we present the sequences of sectors connected by intertwiners $F^{\pm}$ for $N=7$, $q=e^{\frac{2}{3}\pi i}$, $\ell=3$. The red sectors highlight the hidden sectors with non-zero twists and $\ell\mid d$ in the sequences, which explain the multiplicity of the zero twist sectors through the exact sequence construction.}
    \label{fig:diagam_non_n_div}
\end{figure*}
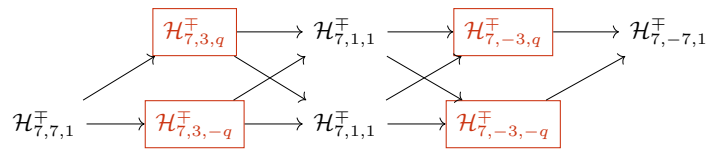

By counting the corresponding sectors on the chain with the same $L(\mathfrak{sl}_2)$ highest weight construction, we arrive at Corollary~\ref{cor:degen-thm-non}. 
The detailed combinatorics argument is spelled out in Appendix~\ref{sec:proof}, and the numerical verification is shown in \refTab{tab:degeneracies} and \refApp{sec:numerics}. 

\section{Conclusion and Outlook}
In this article, we provide a comprehensive algebraic characterization of the degenerate eigenspace at the product state energy level in the one-dimensional periodic XXZ Heisenberg chain with anisotropies corresponding to roots of unity. Our approach connects the established FM string construction to the representation theory of the aTL algebra, and reveals the hidden algebraic structure underlying the enhanced degeneracies. 
For the commensurate chains where $q^N=1$, we prove that the lower bound of the degeneracy at the product state energy grows exponentially in the chain length $N$ by $2^{N/\ell}\ell$ for anisotropies $q^{2\ell}=1$ and smallest such integers $\ell\geq 3$ by exploiting the intertwiners between aTL modules identified in Ref.~\cite{Pinet_2022}.
Further, we extend this framework to the incommensurate chains where $q^N\neq 1$, deriving explicit lower bounds on the degeneracy in Corollary~\ref{cor:degen-thm-non}. Numerical results up to $N=20$ corroborate our equations and suggests that the lower bound for the degeneracy is generally saturated.
We identify a ``hidden half-sequence'' of twisted XXZ chains that underpins the observed degeneracy structure. 
The intertwiners connect sectors with different twists and spins, allowing us to transport eigenvalues across $\stot$ values and establish the relation for the dimension of the degenerate subspace through exact sequences and the loop $\mathfrak{sl}_2$ algebra structure for the $\ell\mid d$ sectors \cite{deguchi_xxz_2007}.
This result is a substantial generalization beyond pre-existing results and demonstrates that the aTL-linear intertwiners naturally link commensurate and incommensurate chains through the same hidden sectors with twisted periodic boundary conditions. The unified treatment clarifies why enhanced degeneracies persist even when $q^N\neq1$, a phenomenon that previously lacked systematic explanations. 

Several directions for future work emerge. First, while our numerical evidence strongly supports the saturation of the lower bounds, we have not proven this hypothesis. Such a proof would likely require a deeper understanding of the composition factors of aTL modules identified at roots of unity in Ref.~\cite{Pinet_2022} to rule out the existence of accidental symmetries. Specifically, one would need to prove the following: For a 1D XXZ chain at root of unity $q$, a sector $(N;d,w)$ contains an eigenstate at product state energy only if $(N;d,w)$ appears in one of the succession sequences starting from fully polarized sectors $(N;N,q^p)$ with $p\in\mathbb Z$. 

Furthermore, product eigenstates and enhanced degeneracies also exist in higher-dimensional lattices \cite{gerken_all_2023,Changlani2018,Miao2025}, which hints at a common symmetry structure that we have yet to identify or connect to our 1D approach based on aTL modules and the Bethe Ansatz. Moreover, exploring whether similar hidden sequences and intertwiner structures exist in other integrable models may reveal broader organizational principles in quantum many-body physics. Finally, investigating the dynamical consequences of these degeneracies, particularly their roles in quantum thermalization and eigenstate thermalization hypothesis violations \cite{Essler:2016ufo,Prosen:2013woz}, represents a promising direction to connect integrable structure and out-of-equilibrium quantum dynamics. 

Our results have implications that extend beyond the specific model considered here. The product eigenstates we study are examples of quantum scars \cite{Andrade:2024iqu,Zhang:2023kxc,popkov_phantom_2021,chandran_quantum_2023,Pizzi:2024urq}, highly atypical states embedded in an otherwise thermal spectrum. Our algebraic characterization of degeneracy at product state energy provides a representation-theoretic understanding of zero-entanglement excitations in integrable systems at special parameter values. Whether analogous structures can be identified in non-integrable models is a natural question, and our framework developed here may help guide that search. The relationship between twisted and periodic boundary conditions through aTL-linear intertwiners points to a broader implication: that degeneracy structures in quantum spin chains can potentially be understood by passing to auxiliary systems with modified boundary conditions. Exploiting this idea in other integrable systems with enhanced symmetries is a natural next step.
Furthermore, the degeneracies we identify are accessible in experiments. Cold-atom and trapped-ion realizations of XXZ chains \cite{jepsen2020SpinTransportTunable,Jepsen:2021gqv} and (linear) quantum computing architectures \cite{vandyke2022PreparingExactEigenstates,jaderberg2025VariationalPreparationNormal,lutz2025AdiabaticQuantumState} provide concrete platforms to test our predictions. 
If realized in solid state systems, the strong sensitivity of the degeneracy on fluctuations in the Heisenberg anisotropy $\Delta$ could enable an application of the 1D Heisenberg chains as quantum sensors as explored for excited states quantum phase transitions \cite{cejnar2021ExcitedstateQuantumPhase}. Yet here, the deviations in $\Delta$ trigger a response in the density of states and Floquet-driven system properties that exponentially grows with the chain length.
More broadly, the way the Bethe ansatz and modern representation theory combine in our analysis highlights the continuing importance of integrable systems as a laboratory for quantum many-body phenomena, giving us both exact results and conceptual frameworks on emergent structures in strongly correlated quantum matter and conformal field theories. 

\begin{acknowledgments}
YH is supported by the Partnership for Innovation, Education and Research (PIER) between Deutsches Elektronen-Synchrotron (DESY) and Universit\"at Hamburg (UHH), as well as the MIT Global Experience (MISTI) program. 
FG and TP acknowledge funding by the European Union (ERC, QUANTWIST, Project number 101039098) and the Cluster of Excellence
‘Advanced Imaging of Matter’ (EXC 2056, project ID 390715994).
The views and opinions expressed
are however those of the authors only and do not
necessarily reflect those of the European Union or
the European Research Council, Executive Agency. The authors are also thankful to Rafael Nepomechie, 
Kenta Suzuki, Elijah Bodish, and Martin Bonkhoff for valuable discussions and insights on the matter.   
\end{acknowledgments}

\bibliographystyle{apsrev4-2}
\bibliography{ref}
\clearpage
\newpage
\onecolumngrid

\appendix
\renewcommand{\theequation}{A.\arabic{equation}}
\setcounter{equation}{0}

\renewcommand{\thefigure}{A.\arabic{figure}}
\setcounter{figure}{0}

\section{Special Cases: $\Delta=0,\pm1$}
\label{sec:xx}
In this appendix, we briefly discuss the structure of product state energy degeneracy for the special anisotropies $\Delta=\pm1,0$ where additional symmetries manifest. 
These results are mostly known but special emphasis on the degeneracy provides a self-contained presentation.

\subsection{$\Delta=1$: XXX Model}
At $\Delta=1$, the system reduces to the XXX model \cite{heisenberg_zur_1928} and the full $\sutwo$ symmetry is restored: the Hamiltonian now commutes with all three generators of the total spin $\mathbf{S}_{\text{tot}}$. Hence, energy eigenstates form irreducible representations of $\sutwo$, labeled by quantum numbers $(S, M)$:
\begin{align}
(\mathbf{S}_{\text{tot}})^2 |S,M\rangle &= S(S+1)|S,M\rangle, \\
\stot |S,M\rangle &= M|S,M\rangle,
\end{align}
where $M \in \{-S, -S+1, \ldots, S-1, S\}$.
The energy depends only on $S$, not on $M$. Therefore, all $(2S+1)$ states in a given irreducible representation are degenerate.
Since the Hamiltonian depends only on the Casimir operator $(\mathbf{S}_{\text{tot}})^2$ 
\begin{equation}
H|S,M\rangle = E(S)|S,M\rangle.
\end{equation}
The fully polarized state $|\uparrow\cdots\uparrow\rangle$ with energy $\varepsilon=N/4$ is a highest-weight vector with the total spin quantum number $S = N/2$. All states in this irreducible representation $\{|N/2, M\rangle : M = -N/2, -N/2+1, \ldots, N/2\}$ are degenerate with degeneracy $2S + 1 = N + 1$. 
To prove that there is no additional degeneracy, the Bethe Ansatz provides the complete spectrum for the XXX model \cite{bethe_zur_1931,fabricius_bethes_2001}. Since the Hamiltonian is a function of the Casimir operator, it assigns different energies to different $S$ sectors. The Bethe Ansatz confirms that states with $S < N/2$ have $E \neq \varepsilon$ (See \refEq{eq:xxz_bae_ener} for more details). Thus, for each $M$, the unique state at $\varepsilon$ is the one with $S = N/2$. Hence, for $\Delta=1$, 
\begin{equation}
    d = N+1.
\end{equation}

\subsection{$\Delta=-1$: Staggered Transformation}
For $\Delta= -1$, one can apply a staggered rotation, a $\pi$ rotation about the $z$-axis
on alternating sites \cite{Takahashi1999,YangYang1966}. This transformation acts as
\begin{align}
\begin{split}
    \text{Odd sites: } &S^x \to S^x, \quad S^y \to S^y, \quad S^z \to S^z, \\
    \text{Even sites: } &S^x \to -S^x, \quad S^y \to -S^y, \quad S^z \to S^z.
\end{split}
\end{align}
For even $N$, this transformation is a unitary equivalence between the periodic $\Delta=-1$ model and the periodic ferromagnetic XXX model, which means the Hamiltonian in Eq.~\eqref{eq:h_1d} with an overall minus sign and $\Delta=+1$ \cite{Takahashi1999}, i.e.,
\begin{equation}
    d=N+1\text{ for even $N$}. 
\end{equation}
For odd $N$, the only product eigenstates are the trivial ones \cite{gerken_all_2023}.
Hence,
\begin{equation}
{d \geq 2 \text{ for odd $N$}}.
\end{equation}

\subsection{$\Delta = 0$: XX Model}

At $\Delta = 0$, the Hamiltonian reduces to the 1D XX model. In this case, the Jordan–Wigner transformation maps the chain to a free fermion model \cite{Jordan1928berDP}: using
\begin{equation}
    \label{eq:jw}
    \begin{aligned}
        & c_j = \sigma_j^-\prod_{k=1}^{n-1}\sigma_k^z, 
    \end{aligned}
\end{equation}
the Hamiltonian becomes
\begin{equation}
    \label{eq:jw,h}
    H = \frac{J}{2}\left[\sum_j\left(c_{j+1}^\dagger c_j + c_j^\dagger c_{j+1} \right)-(-1)^r(c_n^\dagger c_1+c_1^\dagger c_n)\right],
\end{equation}
where $r$ is the total fermion number operator $r=\sum_{j=1}^nc_j^\dagger c_j$.
After the Jorgan-Wigner transformation, we can describe a basis of eigenstates using the angles
\begin{equation}
\label{eq:xx-cosine}
p_i =
\begin{cases}
\frac{2m\pi}{N} &  \text{if } r \text{ odd, }\\
\frac{(2m+1)\pi}{N} & \text{if } r \text{ even,}
\end{cases}
\end{equation}
where $m\in\left\{0,1,\ldots,N-1\right\}$ \cite{deguchi_sl_2_2000}.
The question about the null-space degeneracy can then be rephrased as:
given a set of angles $\{ p_1, p_2, \ldots, p_m \}$ and an integer $M\geq 0$, what is the number of combinations of $M$ different angles such that the sum of their cosines is zero?  
\begin{equation}
\label{eq:xx-cosine-degen}
\sum_{k=1}^r \cos(p_{i_k}) = 0.
\end{equation}
Equivalently, we want to count certain purely imaginary sums of distinct $N^{\text{th}}$ and $2N^{\text{th}}$ roots of unity, which is an unsolved mathematical problem \cite{lenstra1978vanishing,LAM200091,bonacina2023vanishing}, cf. the Online Encyclopeida of Integer Sequences (OEIS) \cite{oeisA107848} and \cite{oeisA103314}.
By brute force combinatorics, the degeneracies of the nullspaces are 
\begin{align}
    \begin{array}{c|c|c|c|c|c|c|c|c|c|c|c|c|c|c|c|c|c|c|c|c|c|c|}
     n
 & 2 & 3 & 4 & 5 & 6 & 7 & 8 & 9 & 10 & 11 & 12 & 13 & 14 & 15 & 16 & 17 & 18 & 19 & 20 & 21 & 22 & \dots
    \\ \hline
    \dim\left(\mathcal{D}\right) & 2 & 2 & 10 & 2 & 14 & 2 & 60 & 20 & 74 & 2 & 386 & 2 & 434 & 346 & 2160 & 2 & 6124 & 2 & 13106 & 2866 & 15554 & \dots
    \end{array},
\end{align}
collected in \cite{oeisA392387}, which are consistent with the directly obtained degeneracies in \refTab{tab:degeneracies}.
According to the product eigenstate formulation in the XXZ model \cite{gerken_all_2023}, $\Delta = 0$ only accommodates two trivial product eigenstates if $N\not\equiv0\mod{4}$, and two additional nontrivial product eigenstates if  $N\equiv0\mod{4}$. The number of non-product states in the nullspace therefore exceeds the product states by far.
For even $N$, there is a known lower bound for the degeneracy of the nullspace \cite{Karle2021_2}
\begin{equation}
    \nullity(H) \geq 2^\frac{N}{2} \ \text{(for $N$ even)},
\end{equation}
originating from the interplay between sublattice symmetry and spatial inversion symmetry.
The relevant sublattice symmetry rotates every second spin about $180^\circ$, i.e., $\left( S^x_i, S^y_i, S^z_i\right) \mapsto \left( -S^x_i, -S^y_i, S^z_i\right)$, for $i$ even, and $H\mapsto-H$.

\section{Review of the Bethe Ansatz and the Transfer Matrix Method}
\label{sec:transfer_matrix}

To calculate the Bethe roots at roots of unity for Fig.~\ref{fig:descendant_tower_b}, we employ the the coordinate Bethe Ansatz and the transfer matrix method. We review the method and our approach here to facilitate reproducibility. 

\subsection{Coordinate Bethe Ansatz}
In order to understand the degeneracy of the XXZ model, we first review the Bethe ansatz which exactly solve the 1D XXZ Heisenberg chain \cite{yang1967,bethe_zur_1931,karbach1998introduction}. For a chain length of $N$, in the sector of $M$ spin-up sites, we denote the state $\ket{n_1,\dots,n_M}$ as the basis state with the spins up at positions $1\leq n_1,\dots,n_M\leq N$ and all other spins down. One writes the coordinate Bethe ansatz \cite{bethe_zur_1931}:
\begin{equation}
\label{eq:xxz_equation}
    \ket{\Psi}=\sum_{1\leq n_1<\dots<n_M\leq N}a(n_1,\dots,n_M)\ket{n_1,\dots,n_M},
\end{equation}
with the wavefunction amplitude
\begin{equation}
\label{eq:xxz_ansatz}
a(n_1,\cdots,n_M)=\sum_{\sigma\in S_M}A_\sigma\exp{\left(i\sum_{j=1}^Mk_{\sigma(j)}n_j\right)}.
\end{equation}
$k_j$ is the quasi-momentum assigned to the $j^{\text{th}}$ magnon (spin-up excitation), which we intend to solve. It is often convenient to parameterize $k_j$ in terms of rapidities $v_j$
\begin{equation}
    \label{eq:rapidity}
e^{ik_j}=\frac{\sinh{(v_j+i\gamma/2)}}{\sinh{(v_j-i\gamma/2)}}.
\end{equation}
$S_M$ is the permutation group on $M$ labels, such that if permutations $a$ and $b$ differs by switching the momenta of two sites $i$ and $j$, i.e. $b=\sigma_{ij}a$, $A_a$ and $A_b$ are related by \cite{wruff2013xxz,bethe_zur_1931}
\begin{align}
    \label{eq:xxz_phase}
&A_b=A_a\exp{(i\zeta_{ij})},\\
        &e^{i\zeta{ij}}\equiv\frac{e^{i(k_i+k_j)}+1-2\Delta e^{ik_i}}{e^{i(k_i+k_j)}+1-2\Delta e^{ik_j}}.
\end{align}
One can now impose the periodic boundary conditions on the 1D XXZ chain and reach the Bethe Ansatz Equations (BAE) for the XXZ model \cite{bethe_zur_1931}:
\begin{equation}
\left(\frac{\sinh{(v_j+i\gamma/2)}}{\sinh{(v_j-i\gamma/2)}}\right)^N=\prod_{k\neq j}^M\frac{\sinh{(v_j-v_k+i\gamma)}}{\sinh{(v_j-v_k-i\gamma)}}.
\end{equation}
Given a set of $M$ Bethe roots $\{v_1,\dots,v_M\}$, the energy of the corresponding state is \cite{miao_q_2021}:
\begin{equation}
\label{eq:xxz_bae_ener}
E=\sum_{j=1}^M\frac{\sin^2{\gamma}}{2\sinh{(v_j+i\gamma/2)}\sinh{(v_j-i\gamma/2)}}.
\end{equation}

\subsection{Transfer Matrix Method}

In the Quantum Inverse-Scattering Method \cite{faddeev1996algebraic}, we employ the transfer matrix to solve the BAE algebraically. We define the $R$-matrix for each pair of the adjacent spins at index $k$ and $k+1$ acting on the tensor product of an auxiliary space
$\mathbb C^{2}$ labeled by $a$ and a quantum space $\mathbb C^{2}$ labeled by $k$:
\begin{equation}
    \label{eq:r-matrix}
    R_{ak}(v)=\begin{pmatrix}
    \sinh(v+i\gamma)& 0& 0& 0\\
    0& \sinh{v}& \sinh{(i\gamma)}&0\\
    0& \sinh{(i\gamma)}&\sinh{v}&0\\
    0& 0& 0& \sinh(v+i\gamma)
\end{pmatrix}.
\end{equation}
This matrix solves the Yang--Baxter equation \cite{yang1967}:
\begin{equation}
    \label{eq:yang-baxter}
    R_{12}(v-u)\,R_{13}(v)\,R_{23}(u)
  =R_{23}(u)\,R_{13}(v)\,R_{12}(v-u).
\end{equation}
We implement the $N$ quantum sites of the entire chain
and form an ordered product
\begin{equation}
    \label{eq:monodromy}
    T_{a}(v)=R_{aN}(v)R_{a,N-1}(v)\cdots R_{a1}(v)
        =\begin{pmatrix}
            A(v)& B(v)\\C(v) &D(v)
        \end{pmatrix},
\end{equation}
where $N$ labels the pair of spins at index $N$ and $1$ as in the periodic boundary condition. The Yang--Baxter equation \refEq{eq:yang-baxter} turns into the RTT relation \cite{faddeev1996algebraic}:
\begin{equation}
    \label{eq:RTT}
    R_{12}(u-v)\,T_{1}(u)\,T_{2}(v)=T_{2}(v)\,T_{1}(u)\,R_{12}(u-v).
\end{equation}
Comparing the matrix entries of both sides, reveals the basic commutation rule
\begin{equation}
    \label{eq:transfer_commutataion}
    \begin{aligned}
        A(u)B(v)&=
  \frac{\sinh(u-v+i\gamma)}{\sinh(u-v)}B(v)A(u)-
  \frac{\sinh (i\gamma)}{\sinh(u-v)}B(u)A(v),
    \end{aligned}
\end{equation}
and an analogous formula with $D$ in place of $A$, which is used to commute $A$ and $D$ past the $B$'s. 

The XXZ Hamiltonian with periodic boundary conditions arises from the first logarithmic derivative at $v=0$:
\begin{equation}
    \label{eq:transfer_hami}
    H_{\text{XXZ}}
  =\frac{\partial}{\partial v}\ln T(v)\Big|_{v=0}
  +\text{constant}.
\end{equation}
Taking the partial trace over the auxiliary index we obtain the transfer matrix:
\begin{equation}
    \label{eq:transfer}
    \mathcal{T}(v)\equiv\mathrm{tr}_{a}T_{a}(v)=A(v)+D(v).
\end{equation}
From \refEq{eq:yang-baxter}, we can find that the transfer matrices commute at different $v$. 

The state corresponding to the Bethe roots $\{v_j\}$ can be generated as
\begin{equation}
\label{eq:bethe_roots}
|\{v_{j}\}_{j=1}^{M}\rangle
  =B(v_{1})\cdots B(v_{M})\ket{\emptyset},
\end{equation}
where $\ket{\emptyset}=\ket{\downarrow\cdots\downarrow}$ is the fully polarized reference state.
The rest of the components of $T_a(v)$ acts on the reference state as
\begin{equation}
    \label{eq:baxter-allup}
     C(v)|\emptyset\rangle=0,
  \quad
  A(v)|\emptyset\rangle=\alpha(v)|\emptyset\rangle,
  \quad
  D(v)|\emptyset\rangle=\delta(v)|\emptyset\rangle,
\end{equation}
where
\begin{equation}
\label{eq:tq-aux}
    \alpha(v)=\sinh^{N}\bigl(v+\tfrac{i\gamma}{2}\bigr),\quad
  \delta(v)=\sinh^{N}\bigl(v-\tfrac{i\gamma}{2}\bigr).
\end{equation}
Acting with $M$ creation operators $B(v_{1})\dots B(v_{M})$ on $\ket{\emptyset}$
and commuting $A$ and $D$ through the $B$'s gives
\begin{equation}
    \label{eq:baxter-commute}
    \begin{aligned}
        \mathcal T(v)|\{v_{j}\}_{1}^{M}\rangle
  =\Bigl[\,\alpha(v)\prod_{j=1}^{M}
            \tfrac{\sinh(v-v_{j}-\gamma)}
                  {\sinh(v-v_{j})}+
          \delta(v)\prod_{j=1}^{M}
            \tfrac{\sinh(v-v_{j}+\gamma)}
                  {\sinh(v-v_{j})}
    \Bigr]
    |\{v_{j}\}_{1}^{M}\rangle+(\text{extra terms}),
    \end{aligned}
\end{equation}
where the extra terms are proportional to terms like
$B(u)B(v_{2})\dots B(v_{M})|\emptyset\rangle$ where $u\neq v_1$ that vanish precisely when each $v_{j}$ obeys the BAE \refEq{eq:xxz_bae}. 
With those constraints satisfied, the coefficient in brackets becomes the eigenvalue $\Lambda(v)$ of the transfer matrix $\mathcal{T}(v)$:
\begin{equation}
    \label{eq:transfer_eigen}
    \begin{aligned}
        \Lambda(v)&=
    \alpha(v)\prod_{j=1}^{M}
       \tfrac{\sinh(v-v_{j}-i\gamma)}{\sinh(v-v_{j})}+\delta(v)\prod_{j=1}^{M}
       \tfrac{\sinh(v-v_{j}+i\gamma)}{\sinh(v-v_{j})}.
    \end{aligned}
\end{equation}
One further defines the Baxter polynomial
\begin{equation}
    \label{eq:baxter_poly}
    Q(v)=\prod_{j=1}^{M}\sinh(v-v_{j}).
\end{equation}
One can then arrive at the {Baxter's $TQ$ functional relation} \cite{faddeev1996algebraic,baxter1982exactly}, which simplifies the BAE into a linear equation. The roots of the polynomial $Q(v)$ are the Bethe roots of the system. 
\begin{equation}
\label{eq:functional}
\begin{aligned}
    &(-1)^M\Lambda(v)Q(v)=\sinh^N\left(v-\frac{i\gamma}{2}\right)Q(v+i\gamma)+\sinh^N\left(v+\frac{i\gamma}{2}\right)Q(v-i\gamma).
\end{aligned}
\end{equation}

In particular, Ref.~\cite{miao_q_2021} has derived the form of the $Q$ polynomials at roots of unity. At roots of unity, the form of the Baxter polynomial $Q(v)$ depends on whether one starts from the fully polarized state or from a primitive state. In the first case, beginning with the fully polarized state, the $Q$-polynomial is given explicitly by Ref.~\cite{miao_q_2021}:
\begin{equation}
    Q(v,\phi)
=\sum_{k=0}^{\ell_2-1}\Big(q^{-k-\tfrac12}t-q^{k+\tfrac12}t^{-1}\Big)^{N} e^{ik\phi},
\quad t=e^{v}.
\end{equation}
This form arises directly from Baxter’s TQ relation when $Q$ is generated by the fully polarized vector: for special values of the twist $\phi$ satisfying $1-q^{\ell N} e^{i\ell\phi}=0$ with $q^{\ell}=\pm 1$, one finds the periodicity condition $Q_(v+\eta,\phi)=e^{i\phi}Q(v,\phi)$ \cite{miao_q_2021}. The roots of the $Q$ polynomials are then the FM roots for this descendant tower.

In contrast, for states that descend from primitive roots, for the primitive state to be at the product state energy level, the primitive roots are all at infinity. Ref.~\cite{miao_q_2021} constructs the decomposition
\begin{equation}
    Q(v)=Q_r(v)Q_s(v)t^{n_{-\infty}-n_{+\infty}}, \quad t=e^{v},
\end{equation}
with contributions from regular roots $Q_r$, from finite-modulus FM strings $Q_s$, and from the multiplicities of $\pm\infty$ roots. $n_{\pm\infty}$ is the number of primitive roots at $\pm\infty$ respectively. 
The FM string part takes the form \cite{miao_q_2021}
\begin{equation}
    Q_s(v)\propto \prod_{m=1}^{n_\text{FM}}\left(t^{\ell}-e^{2\ell\alpha_m}\right),
\end{equation}
where $\alpha_m$ is the FM string center. For only infinite primitive roots and $\phi=0$, $Q_r\equiv 1$. The FM string centers are then found by solving the following quantization relation \cite{miao_q_2021} derived from Baxter's TQ relation:
\begin{equation}
    \sum_{k=0}^{\ell-1}\sinh^N\big(v+(k+\tfrac12)i\gamma\big)
e^{-(2k+1)(n_{-\infty}-n_{+\infty})\gamma}=0.
\end{equation}
We additionally verify the numerically obtained FM string center with a numerical solver following the algorithm of Refs.~\cite{fabricius_bethes_2001,miao_q_2021}.

\subsection{Connection between the Intertwiners and the Bethe Ansatz}
\label{appendix:connection_FM_aTL}

Ref.~\cite{Pinet_2022} works out the explicit forms of the intertwiners
\begin{align}
    &F^\mp_{(d,w);(t,v)}:\mathcal{H}_{N;t,v}^\mp\to\mathcal{H}_{N;d,w}^\mp\qquad\text{where}\qquad|\sigma_1\dots \sigma_N\rangle_v^\pm\mapsto F_\pm^{(n)}|\sigma_1\dots \sigma_N\rangle_{w}^\pm,\\
    &E^\pm_{(t,v);(d,w)}: \mathcal{H}_{N;d,w}^\pm\to\mathcal{H}_{N;t,v}^\pm\qquad\text{where}\qquad
|\sigma_1\dots \sigma_N\rangle_{w}^\pm\mapsto E_\pm^{(n)}|\sigma_1\dots \sigma_N\rangle_v^\pm.
\end{align}
Here, $|\sigma_1\dots \sigma_N\rangle_{w}$ indicates states associated with the Hamiltonian with twisted boundary condition \refEq{eq:twisted_bc} with twist $w$ where $\sigma_j=\pm$, and the $E_\pm^{(n)}$ and $F_\pm^{(n)}$ act as
\begin{align}
F^{(n)}_{\pm}|\sigma_1\dots \sigma_N\rangle_w^{\pm} &= \sum_{\substack{1\leq j_1 < \dots  < j_n \leq N\\ \sigma_{j_1} = \dots  = \sigma_{j_n} = +}} \Big(\prod_{k=1}^n\prod_{j=j_{k}+1}^{j_{k+1}-1} q^{\pm k\sigma_j}\Big)|\sigma_1\dots \sigma_{j_1-1}(-)\sigma_{j_1+1}\dots \sigma_{j_2-1}(-)\sigma_{j_2+1}\dots \sigma_N\rangle_w^{\pm}
\end{align}
and
\begin{align}
E^{(n)}_{\pm}|\sigma_1\dots \sigma_N\rangle_v^{\pm} &= \sum_{\substack{1\leq j_1 < \dots  < j_n \leq N\\ \sigma_{j_1} = \dots  = \sigma_{j_n} = -}}\Big(\prod_{k=0}^{n-1}\prod_{j=j_k+1}^{j_{k+1}-1}q^{\pm(k-n)\sigma_j}\Big)|\sigma_1\dots \sigma_{j_1-1}(+)\sigma_{j_1+1}\dots \sigma_{j_2-1}(+)\sigma_{j_2+1}\dots \sigma_N\rangle_v^{\pm}.
\end{align}
These are Lusztig's divided powers \cite{LUSZTIG1988}. 
One might naively expect the total spin operators $S^\pm_{\text{tot}} = \sum_j S^\pm_j$ to be the relevant raising and lowering operators for the XXZ chain, as they are in the isotropic Heisenberg case.
Yet, neither they nor their $q$-analogs mediate between degenerate states directly, since the XXZ Hamiltonian with periodic boundary conditions does not possess a $U_q(\mathfrak{sl}_2)$ symmetry. Yet, at commensurate roots of unity, it can be shown that the operators $S^{\pm({\ell})}$, i.e., the ${\ell}^\text{th}$ divided powers of the quantum group generators
\begin{align}
S^{\pm({\ell})}  = \sum_{1 \le j_1 < \cdots < j_{\ell} \le {N}} &\underbrace{{q}^{N S^z} \otimes \cdots
\otimes {q}^{N S^z}}_{j_1-1 \text{ times}} \otimes S^{\pm} \otimes \underbrace{{q}^{(N-2) S^z} \otimes  \cdots \otimes {q}^{(N-2) S^z}}_{j_2 - j_1 -1 \text{ times}}  \otimes
\nonumber \\
&S^{\pm} \otimes {q}^{(N-4) S^z} \otimes \cdots \otimes {q}^{S^z} \otimes S^{\pm} \otimes \underbrace{{q}^{-{N} S^z} \otimes \cdots \otimes {q}^{-{N} S^z}}_{N - j_\ell \text{ times}},
\label{eqn:deguchi_string}
\end{align}
commute with the Hamiltonian in the sector where $\ell\mid d$, see Eq.~(1.2) in~\cite{deguchi_xxz_2007}.
In this equation, the operators $S^\pm$ and $S^z$ act on the local Hilbert space.
Note that a prefactor that is zero at the $\ell^\text{th}$ roots of unity has been lifted, a construction going back to Lusztig~\cite{LUSZTIG1988} and applied to spin chains soon after~\cite{PASQUIER1990523}.
These operators $S^{\pm({\ell})}$ coincide precisely with the intertwiners $E^{({\ell})}_\pm$ and $F^{({\ell})}_\pm$ in Lemma~4.13 of Ref.~\cite{Pinet_2022}, up to conventions. Since $S^{\pm({\ell})}$ commutes with the Hamiltonian and raises/lowers $S^z$ by ${\ell}$, it maps a Bethe eigenstate with $m$ magnons to another eigenstate with the same energy but $m\pm {\ell}$ magnons by adding an FM string, whose contribution to the energy vanishes at roots of unity. Thus, the intertwining operators $E^{({\ell})}$ and $F^{({\ell})}$ are precisely the generators that add or remove complete FM strings from Bethe states, and together with the conjugates $T^{\pm({\ell})} = (S^{\pm({\ell})})^*$ they generate the $L(\mathfrak{sl}_2)$ that underlies the enhanced degeneracy structure at roots of unity.
Hence, the intertwiner $F^\mp_{(d,w);(t,v)}$ and $E^\pm_{(t,v);(d,w)}$ are generalizations of FM strings across sectors with different twists. 

\section{Detailed Proofs of the degeneracies}\label{sec:proof}
This appendix provides our detailed proofs of the results in the main text.

\begin{proof}[Proof of Lemma~\ref{lem:transport}]
For notation simplicity, we label $\mathcal{H}_{N;t,x}=:\mathcal H_1$, $\mathcal{H}_{N;d,z}=:\mathcal H_2$. 
Let the morphism be $\mu:\mathcal{H}_1\to\mathcal{H}_2$. We treat $H_\text{XXZ}$ as an element of the aTL$_N$ algebra, and the explicit form of $H_\text{XXZ}$ as its representations on the corresponding Hilbert spaces $\mathcal H_1$ and $\mathcal H_2$. As $\mu$ is aTL$_N$-linear, we can write the representations of aTL$_N$ on the two modules: 
\begin{equation}
    \begin{aligned}
        &\pi_1:\mathrm{aTL}_N\to\mathrm{End}(\mathcal{H}_1),\\
        &\pi_2:\mathrm{aTL}_N\to\mathrm{End}(\mathcal{H}_2).
    \end{aligned}
\end{equation}
Then we have
\begin{equation}
    \mu\pi_1(a)v_1=\pi_2(a)\mu v_1\quad\forall a\in\mathrm{aTL}_N,\;v_1\in\mathcal{H}_1.
\end{equation}
In particular, as $H_\text{XXZ}\in\mathrm{aTL}_N$, if $v_1$ is an eigenvector of the representation of $H_\text{XXZ}$ on $\mathcal{H}_1$ such that
\begin{equation}
    \pi_1(H_\text{XXZ})v_1=\lambda v_1,
\end{equation}
then we can commute $\lambda\in\mathbb{R}$ across $\mu$ and obtain
\begin{equation}
    \pi_2(H_\text{XXZ})\mu v_1=\lambda \mu v_1.
\end{equation}
Hence $\mu v_1\in\mathcal{H}_2$ is an eigenvector of the representation of $H_\text{XXZ}$ on $\mathcal{H}_2$ with the same eigenvalue $\lambda$. 

\end{proof}

\begin{proof}[Proof of Lemma~\ref{lemma:restriction}]
Since each $d_i$ is $G$-linear and $H\in G$, we have $\rho_{V_{i+1}}(H)\circ d_i = d_i\circ \rho_{V_i}(H)$ for all $i$. Because $\rho_{V_i}(H)$ is diagonalizable, the spectral projector $P_{i,h}$ onto $V_{i,h}$ is a polynomial in $\rho_{V_i}(H)$; hence $d_i P_{i,h} = P_{i+1,h} d_i$.
If $v\in V_{i,h}$ then $\rho_{V_{i+1}}(H)(d_i v)=d_i(\rho_{V_i}(H)v)=hd_i v$, so $d_i$ restricts to $d_{i,h}:V_{i,h}\to V_{i+1,h}$. For $0<i<n$,
$\operatorname{im}(d_{i-1,h})\subseteq\ker(d_{i,h})$ since $d_id_{i-1}=0$. Conversely, if $x\in\ker(d_{i,h})$, exactness gives $y\in V_{i-1}$ with $d_{i-1}y=x$, and then
\begin{equation}
    d_{i-1}(P_{i-1,h}y)=P_{i,h}d_{i-1}y=P_{i,h}x=x,
\end{equation}
so $x\in\operatorname{im}(d_{i-1,h})$. Injectivity at the left and surjectivity at the right are preserved by the same projection argument.
\end{proof}

\begin{proof}[Proof of Corollary~\ref{cor:degen-thm-non}]
 We split the proof in two cases: if $q^\ell=1$, and if $q^\ell=-1$. 

\paragraph{$q^\ell=1$.} Here $\ell$ is odd.
We know that for $N=k\ell$, where $k\in\mathbb{N}$, the sequences of sectors are the ones in Fig. \ref{fig:diagram_oddl}. For generic $N$ not divisible by $\ell$, write $N=k\ell-2r$, where $k,r\in\mathbb{Z}$. In order for the extra degeneracy at the product state energy level to exist, the sequences of sectors connected by the intertwiner have to start with $(N;N,m)$ for some $m$ and have sectors of zero twist in the sequences. In particular, we can see that the only sequence that exists for generic $N=k\ell-2r$ not divisible by $\ell$ is the truncated sequence from $(k\ell;k\ell,q^r)$, from which we take all the sectors that evolve from $(k\ell-2r,1)$. We show the sequence in Fig. \ref{fig:diagram_oddl_nondiv}. 
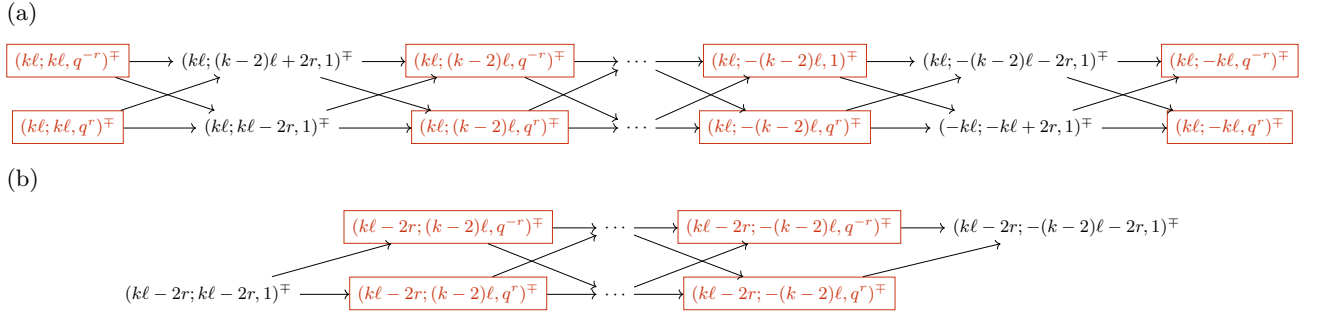
\begin{figure*}[]
    \centering
    \begin{subfigure}{\textwidth}
        \caption{}
        \centering
        \adjustbox{scale=0.75,center}{
        \begin{tikzcd}
          |[draw=cbRed]| \color{cbRed}{(k\ell;k\ell,q^{-r})^\mp}
            \arrow[dr]
             \arrow[r]
          &
          (k\ell;(k-2)\ell+2r,1)^\mp
            \arrow[dr]
             \arrow[r]
          &
          |[draw=cbRed]| \color{cbRed}{(k\ell;(k-2)\ell,q^{-r})^\mp}
            \arrow[dr]
             \arrow[r]
          &
          \cdots
            \arrow[dr]
             \arrow[r]
          &
          |[draw=cbRed]| \color{cbRed}{(k\ell;-(k-2)\ell,1)^\mp}
            \arrow[dr]
             \arrow[r]
             &
        {(k\ell;-(k-2)\ell-2r,1)^\mp}
            \arrow[dr]
             \arrow[r]
             &
          |[draw=cbRed]| \color{cbRed}{(k\ell;-k\ell,q^{-r})^\mp}
          \\
          |[draw=cbRed]| \color{cbRed}{(k\ell;k\ell,q^{r})^\mp}
          \arrow{ur}
          \arrow[r]
          &
          (k\ell;k\ell-2r,1)^\mp
          \arrow{ur}
          \arrow[r]
          &
          |[draw=cbRed]| \color{cbRed}{(k\ell;(k-2)\ell,q^{r})^\mp}
          \arrow{ur}
          \arrow[r]
          &
          \cdots
          \arrow{ur}
          \arrow[r]
          &
          |[draw=cbRed]| \color{cbRed}{(k\ell;-(k-2)\ell,q^{r})^\mp}
          \arrow{ur}
          \arrow[r]
           &
        {(-k\ell;-k\ell+2r,1)^\mp}
          \arrow{ur}
          \arrow[r]
           &
          |[draw=cbRed]| \color{cbRed}{(k\ell;-k\ell,q^{r})^\mp}
        \end{tikzcd}
        }
        \label{fig:diagram_oddl_nondiv_a}
    \end{subfigure}
    \\[0.5em]
    \begin{subfigure}{\textwidth}
        \caption{}
        \centering
        \adjustbox{scale=0.75,center}{
        \begin{tikzcd}
          &
          |[draw=cbRed]| \color{cbRed}{(k\ell-2r;(k-2)\ell,q^{-r})^\mp}
            \arrow[dr]
             \arrow[r]
          &
          \cdots
            \arrow[dr]
             \arrow[r]
          &
          |[draw=cbRed]| \color{cbRed}{(k\ell-2r;-(k-2)\ell,q^{-r})^\mp}
             \arrow[r]
             &
        {(k\ell-2r;-(k-2)\ell-2r,1)^\mp}
          \\
          (k\ell-2r;k\ell-2r,1)^\mp
          \arrow{ur}
          \arrow[r]
          &
          |[draw=cbRed]| \color{cbRed}{(k\ell-2r;(k-2)\ell,q^{r})^\mp}
          \arrow{ur}
          \arrow[r]
          &
          \cdots
          \arrow{ur}
          \arrow[r]
          &
          |[draw=cbRed]| \color{cbRed}{(k\ell-2r;-(k-2)\ell,q^{r})^\mp}
          \arrow{ur}
        \end{tikzcd}
        }
        \label{fig:diagram_oddl_nondiv_b}
    \end{subfigure}
    \caption{The truncation of sequences for $q^\ell=1$ and $\ell\nmid N$. (\subref{fig:diagram_oddl_nondiv_a}) The full sequence for $N=k\ell$ that starts at twist $\log(q^r)$. (\subref{fig:diagram_oddl_nondiv_b}) The truncated sequence for $N=k\ell-2r$ that starts with $N=k\ell-2r$ and twist $0$. }
    \label{fig:diagram_oddl_nondiv}
\end{figure*}
We hence just need to count the multiplicity from this truncated sequence. Since all the sectors with $\ell\mid d$ there start with $(k-2)\ell$, the multiplicity from the $L(\mathfrak{sl}_2)$ highest weight vector is $2^{k-1}$. Relabeling $k$ with $N$, we arrive at the multiplicity bound of $2^{2\lfloor\frac{N}{2\ell}+\frac{1}{2}\rfloor}$ for odd $N$ and $2^{2\lfloor\frac{N}{2\ell}\rfloor+1}$ for even $N$. 

\paragraph{$q^\ell=-1$.}
Here, the situation splits into the cases where $N$ is odd or even. If $N$ is even, the truncated sequences behave as Fig. \ref{fig:diagram_ql-1_evenn_nondiv} to give the multiplicity of $2^{2\lfloor\frac{N}{2\ell}\rfloor+1}$.  
If $N$ is odd, $\ell$ is odd, and Lemma \ref{lem:oddl_oddN} shows that the zero-twist sectors do not have a successor. Therefore, the total multiplicity is just the spin-flip multiplicity of $2$, same as the multiplicity at generic anisotropies.     
\end{proof}

\begin{figure*}[]
    \centering
 \adjustbox{scale=1.0,center}{
\begin{tikzcd}
  &
  |[draw=cbRed]| \color{cbRed}{(N;2(k-1)\ell,-q^{-r})^\mp}
    \arrow[dr]
     \arrow[r]
  &
    |[draw= cbBlue, dashed]| \color{cbBlue}{(N;2(k-2)\ell+2r,-1)^\mp}
    \arrow[dr]
     \arrow[r]
  &
  |[draw=cbRed]| \color{cbRed}{(N;2(k-2)\ell,q^{-r})^\mp}
    \arrow[dr]
     \arrow[r]
  &
  \cdots
  \\
  (N;2k\ell-2r,1)^\mp
  \arrow{ur}
  \arrow[r]
  &
  |[draw=cbRed]| \color{cbRed}{(N;2(k-1)\ell,-q^{r})^\mp}
  \arrow{ur}
  \arrow[r]
  &
    {(N;2(k-1)\ell-2r,1)^\mp}
  \arrow{ur}
  \arrow[r]
  &
  |[draw=cbRed]| \color{cbRed}{(N;2(k-2)\ell,q^{r})^\mp}
  \arrow{ur}
  \arrow[r]
  &
  \cdots
\end{tikzcd}
}
 \adjustbox{scale=1.0,center}{
\begin{tikzcd}
  &
 |[draw=cbRed]| \color{cbRed}{(N;2(k-1)\ell,q^{-r})^\mp}
    \arrow[dr]
     \arrow[r]
  &
  {(N;2(k-2)\ell+2r,1)^\mp}
    \arrow[dr]
     \arrow[r]
  &
  |[draw=cbRed]| \color{cbRed}{(N;2(k-2)\ell,-q^{-r})^\mp}
    \arrow[dr]
     \arrow[r]
  &
  \cdots
  \\
 |[draw= cbBlue, dashed]| \color{cbBlue}{(N;2k\ell-2r,-1)^\mp}
  \arrow{ur}
  \arrow[r]
  &
  |[draw=cbRed]| \color{cbRed}{(N;2(k-1)\ell,q^{r})^\mp}
  \arrow{ur}
  \arrow[r]
  &
   |[draw= cbBlue, dashed]| \color{cbBlue}{(N;2(k-1)\ell-2r,-1)^\mp}
  \arrow{ur}
  \arrow[r]
  &
  |[draw=cbRed]| \color{cbRed}{(N;2(k-2)\ell,-q^{r})^\mp}
  \arrow{ur}
  \arrow[r]
  &
  \cdots
\end{tikzcd}
}
    \caption{The two truncated sequences for $q^\ell=-1$, $\ell\nmid N$, and even $N$. These two sequences are the sequences that contain sectors of zero twist.}
    \label{fig:diagram_ql-1_evenn_nondiv}
\end{figure*}
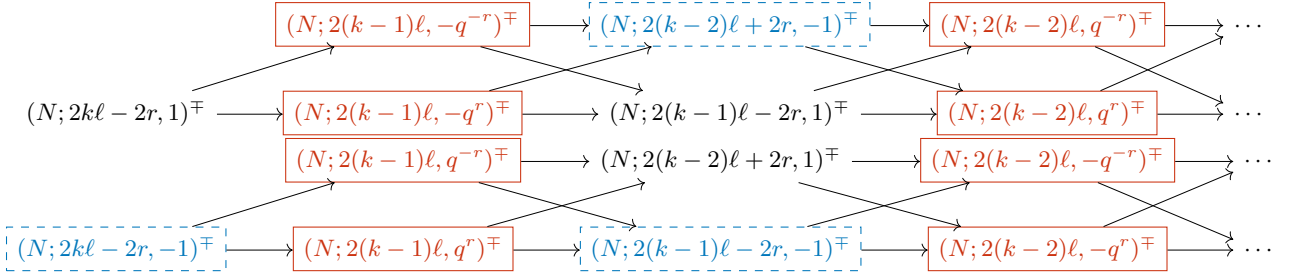

\section{Numerically determining the degeneracy at product state energy}
\label{sec:numerics}

This section explains how we obtain the degeneracies of the periodic XXZ Heisenberg chain in \refEq{eq:h_1d} at product state energy $\varepsilon$ and roots of unity, which we collect in \refTab{tab:degeneracies}. To this end we determine the nullity of $H - \varepsilon,$ with $\varepsilon=\Delta N /4$ in \refEq{eq:product_state_energy}, which we obtain by summing the nullities of its projections onto the $S_z$ eigenspaces. Calculations are sped up by the spin flip symmetry, see the dashed equatorial lines in \refFig{fig:descendant_tower_a} and \refFig{fig:diagram_twolines_b}. 
To this end, we collect random nullspace vectors and calculate the dimension of their span by the Gram-Schmidt process, where we neglect basis vectors with absolute values smaller than $10^{-5}$. To ensure convergence, we demand the sample size to be four times larger than the obtained basis size.

To construct random nullspace vectors of a $d \times d$ matrix $M$, we first choose a uniformly randomly distributed normal vector $b\in\mathbb{C}^{d}$ and construct the vector $y_b$ in the image of $M$ that is closest to $b$, i.e.,  $|y_b-b|$ is minimal. Then $n = b - y_b$ is the projection of $b$ onto the nullspace.
We find $y_b$ by minimizing $|Mx_b-b|$, which is equivalent to solving $M^2 x_b = M b$ with $x \in \mathbb{C}^d$. Here, we use the conjugate gradient method \cite{Barrett1994IterativeMethods} with accepted remaining absolute error of $10^{-11}$ and without preconditioner. We then set $y_b = M x$. For the problem at hand, we find that this method is more stable than solving $M y_b = 0$, possibly because of the latter's ambiguous solution for $y_b$. We further find no numerical advantage in using the minimal residual method or Jacoby or Chebychev preconditioners \cite{Barrett1994IterativeMethods}. The bottleneck of the approach is the exponentially increasing density of states around zero energy with chain length $N$, and the thereby reduced convergence of the minimization methods, which renders these and other methods, like the ones including matrix product states, imprecise.
\end{document}